\documentclass[lettersize,journal]{IEEEtran}
\usepackage{amsmath,amsfonts}
\usepackage{bm}
\usepackage{amssymb}
\usepackage{amsthm}
\usepackage{algorithmic}
\usepackage{algorithm}
\usepackage{array}
\usepackage{textcomp}
\usepackage{threeparttable}
\usepackage{stfloats}
\usepackage{url}
\usepackage{pifont}
\usepackage{verbatim}
\usepackage{graphicx}
\usepackage{multirow}
\usepackage{colortbl}
\usepackage[table,xcdraw]{xcolor}
\usepackage{cite}
\usepackage{xcolor}
\usepackage{CJK}
\usepackage{pifont}
\usepackage{booktabs}
\usepackage[]{float}
\usepackage{hyperref}
\usepackage{subfigure}
\usepackage{subeqnarray}
\usepackage{cases}
\usepackage{diagbox}
\usepackage{color}
\usepackage{tabularray}
\newtheorem{lemma}{Lemma}
\newtheorem{theorem}{Theorem}
\newtheorem{remark}{Remark}
\usepackage{tabu}
\usepackage{makecell}
\begin{document}
\title{CovertAuth: Joint Covert Communication and Authentication in MmWave Systems}

\author{Yulin~Teng, Keshuang~Han,
Pinchang~Zhang,~\IEEEmembership{Member,~IEEE},
Xiaohong Jiang,~\IEEEmembership{Senior Member,~IEEE},\\
Yulong Shen,~\IEEEmembership{Member,~IEEE},
and Fu~Xiao,~\IEEEmembership{Senior Member,~IEEE}
\thanks{Y. Teng, K. Han, and P. Zhang are with the School of Computer Science, Nanjing University of Posts and Telecommunications, Nanjing, Jiangsu, 210023, China (email:\{2023040511, 1023041139, zpc\}@njupt.edu.cn). They are also with the State Key Laboratory of Integrated Services
Networks (Xidian University).}
\thanks{X. Jiang is with the School of Systems Information Science, Future University Hakodate, Hakodate, 041-8655, Japan (e-mail: jiang@fun.ac.jp).}
\thanks{Y. Shen is with the School of Computer Science and Technology, Xidian
University, Xi’an 710071, China (e-mail: ylshen@mail.xidian.edu.cn).}
\thanks{F. Xiao is with the School of Computer Science,
Nanjing University of Posts and Telecommunications, Nanjing, Jiangsu,
210023, China (e-mail: xiaof@njupt.edu.cn).}
}

\markboth{Journal of \LaTeX\ Class Files,~Vol.~14, No.~8, August~2021}%
{Shell \MakeLowercase{\textit{et al.}}: A Sample Article Using IEEEtran.cls for IEEE Journals}

\maketitle
\begin{abstract}
Beam alignment (BA) is a crucial process in millimeter-wave (mmWave) communications, enabling precise directional transmission and efficient link establishment. However, due to characteristics like omnidirectional exposure and the broadcast nature of the BA phase, it is particularly vulnerable to eavesdropping and identity impersonation attacks. To this end, this paper proposes a novel secure framework named CovertAuth, designed to enhance the security of the BA phase against such attacks. 
In particular, to combat eavesdropping attacks, the closed-form expressions of successful BA probability and covert transmission rate are first derived. Then, a covert communication problem aimed at jointly optimizing beam training budget and transmission power is formulated to maximize covert communication rate, subject to the covertness requirement. 
An alternating optimization algorithm combined with successive convex approximation is employed to iteratively achieve optimal results.
To combat impersonation attacks, the mutual coupling effect of antenna array impairments is explored as a device feature to design a weighted-sum energy detector-based physical layer authentication scheme. Moreover, theoretical models for authentication metrics like detection and false alarm probabilities are also provided to conduct performance analysis. Based on these models, an optimization problem is constructed to determine the optimal weight value that maximizes authentication accuracy. 
Finally, simulation results demonstrate that CovertAuth presents improved detection accuracy under the same covertness requirement compared to existing works. 

\end{abstract}
\vspace{-1em}
\section{Introduction}
\subsection{Background and Motivation}
\IEEEPARstart{B}{enefit} from the abundant spectrum resources and massive antenna arrays, millimeter-wave (mmWave) communication can exploit the highly directional beamforming technology to acquire ultra-high data transmission rates and extremely low latency \cite{DBLP:journals/comsur/TanLGWPZZL24}. Beam alignment (BA) is a critical prerequisite for the effective implementation of beamforming in mmWave communication. Accurate BA can significantly improve signal reception quality, compensate for high path loss, and maximize beamforming gain.
However, the inherent characteristics of the BA process such as omnidirectional exposure, high-probability line-of-sight channel, and broadcast nature of training sequences also make it particularly vulnerable to identity impersonation and eavesdropping attacks \cite{DBLP:conf/wisec/SteinmetzerYH18}. 
Specifically, omnidirectional or quasi-omnidirectional scanning is generally used to find the optimal transmission path during the BA phase \cite{DBLP:journals/twc/YangZTS24}, inadvertently increasing exposure to potential adversaries. By transmitting deceptive alignment signals that mimic legitimate ones, adversaries can mislead the base station into aligning their beams with a malicious source \cite{DBLP:conf/mobicom/SteinmetzerAAH018}.
Furthermore, the broadcast nature of beam training sequences during the BA stage allows adversaries to readily capture and analyze these sequences to infer beam directions and the underlying signal frame structure \cite{DBLP:journals/tifs/QiuCZ23}. 
Once the beam direction is grasped, adversaries can position themselves strategically to monitor the communication process and even locate the location of the transmitter. 
Consequently, it is essential to design an effective and reliable defense mechanism for the BA stage to secure the establishment of the mmWave communication link. 

Recently, physical layer security (PLS) technique has been considered a cost-effective approach to combat impersonation and eavesdropping attacks with reduced resource demands and higher architecture compatibility \cite{DBLP:journals/tmc/XieZZTLN24}, \cite{DBLP:journals/ton/YuYL23}. 
On the one hand, PLS for authentication utilizes inherent wireless channel characteristics or device-specific hardware fingerprints to validate user/device legitimacy, \cite{DBLP:journals/ton/XieTHL21}. On the other hand, PLS exploits the randomness of wireless channels such as fading, noise, and interference to improve communication secrecy \cite{DBLP:journals/ton/LiSKWLLZ23}. 
However, in certain critical scenarios like military combat and location tracking service, merely protecting content confidentiality is far insufficient. It is also imperative to conceal the existence of transmission activities to effectively mitigate potential threats.
Physical layer covert communication is an emerging PLS technique that allows the transmitter to discreetly convey information to the intended receiver at a specific rate, remaining undetected by adversaries and fundamentally preventing eavesdropping \cite{DBLP:journals/pieee/JiangWCS24}.
Therefore, it is expected that physical layer authentication combined with covert communication techniques can provide novel insights for designing secure mechanisms in the BA phase of mmWave systems.

Several recent studies have focused on deploying PLS methods in mmWave systems to address identity-based impersonation \cite{DBLP:conf/infocom/WangJWLZ20}, \cite{DBLP:journals/tifs/NosouhiSGD22}, \cite{DBLP:journals/tifs/LuLSL23}, \cite{DBLP:journals/tifs/TengZCJX24}, or eavesdropping attacks \cite{DBLP:journals/tifs/ZhangLZJX22},  \cite{DBLP:journals/twc/JamaliM22}, \cite{DBLP:journals/twc/WangLN22}, \cite{DBLP:journals/twc/XiaoHLWSWY24}. For the physical layer authentication approaches,  Wang \textit{et. al} exploit the signal-to-ratio (SNR) trace in the sector-level sweep process to design an efficient spoofing attack detection scheme for IEEE 802.11ad mmWave networks \cite{DBLP:conf/infocom/WangJWLZ20}. Following this line, the work in \cite{DBLP:journals/tifs/NosouhiSGD22}  investigates the uniqueness of the mmWave beam pattern feature and utilizes a deep autoencoder for illegal device detection with 98.6\% accuracy. Later, Lu \textit{et. al} leverage channel phase responses to mitigate performance losses due to imperfect channel correlation, demonstrating superior authentication performance even under low SNR conditions \cite{DBLP:journals/tifs/LuLSL23}. For the covert communication schemes, the authors of \cite{DBLP:journals/twc/JamaliM22} employ a dual-beam mmWave transmitter to concurrently transmit desired signal to the legitimate receiver and a jamming signal to interfere with the adversary's detection capability. In \cite{DBLP:journals/twc/WangLN22}, a full-duplex receiver is considered to 
generate jamming signals with a time-varying power for masking the presence of the legitimate transmitter. Xiao \textit{et. al} explore the use of simultaneously transmitting and reflecting reconfigurable intelligent surfaces to support covert communication in a mmWave system \cite{DBLP:journals/twc/XiaoHLWSWY24}. However, the above methods are primarily designed to address one specific security threat, and their reliability may degrade when facing adversaries possessing multiple attack capabilities.

\begin{table*}
\vspace{-2em}
\centering
\caption{Comparison with Existing Physical Layer Security Works}
\label{comparison}
\begin{threeparttable}
\begin{tblr}{
	  width = \linewidth,
	  colspec = {Q[c]Q[290]Q[290]Q[c]Q[42]Q[42]Q[580]},
	  cell{1}{1} = {r=2}{},
	  cell{1}{2} = {r=2}{},
	  cell{1}{3} = {r=2}{},
	  cell{1}{4} = {c=3}{0.126\linewidth},
	  cell{1}{7} = {r=2}{},
	  vlines,
	  hline{1,11} = {-}{0.08em},
	  hline{2} = {4-6}{},
	  hline{3-10} = {-}{},
	}
\textbf{Works}& \textbf{~~~~~~System Model} & \textbf{~~~Performance Metrics} & \textbf{Security Guarantee} & & & \textbf{~~~~~~~~~~~~~~~~~~Main Contributions} \\
&      &     & Co.\tnote{1}   & Au. &  Pr.  &   \\
{}\cite{DBLP:conf/infocom/WangJWLZ20} & Data transmission scenario of mmWave MIMO system  & Detection probability                     & \ding{55}                       & $\checkmark$   & {
	  \ding{55}} & An authentication scheme exploiting the
	  SNR trace during the sector-level sweep process to identify the transmitter. \\
{}\cite{DBLP:journals/tifs/TengZCJX24} & Data transmission scenario of
	  UAV-Ground system                     & {
	  Detection and false alarm probabilities
	  }                   & \ding{55}                       & $\checkmark$   & \ding{55}            & Authentication framework based on
	  Laplace distribution modeling of channel sparsity.                                                                                                \\
{}\cite{DBLP:journals/tifs/ZhangLZJX22} & Beam training and data transmission
	  scenarios of mmWave MIMO system & Covert communication rate                                                         & $\checkmark$                       & \ding{55}   & $\checkmark$            & Covert beam training strategy for
	  multi-user communication to optimize covert throughput.                                                                                           \\
{}\cite{DBLP:journals/twc/WangLN22} & Data transmission scenario of
	  mmWave full duplex system           & Detection error probability and covert
	  communication rate                       & $\checkmark$                       & \ding{55}   & $\checkmark$            & Full-duplex covert mmWave communication
	  with joint optimization of beamforming, transmit power, and jamming to
	  maximize covert rate. \\
{}\cite{DBLP:journals/jsac/XieZCT22} & Data transmission scenario of NOMA system
                & ~Detection and outrage probabilities
 & $\checkmark$                       & \ding{55}   & $\checkmark$            & Channel-based key generation for information confidentiality and a hybrid authentication leveraging channel characteristics and tag-based verification.\\
Ours & Beam training scenario of mmWave MIMO system
                & Detection, false alarm and  beam alignment probabilities, covert communication rate
 & $\checkmark$   & $\checkmark$   & $\checkmark$            & CovertAuth utilizes mutual coupling for identity validation and integrates it with a covert communication optimization framework to combat eavesdropping and impersonation attacks. 
\end{tblr}
\begin{tablenotes}    
    \footnotesize               
    \item[1] Co.: confidentiality,    ~~~Au.: authenticity,     ~~~Pr.: privacy.              
\end{tablenotes}            
\end{threeparttable}     
\vspace{-2em}
\end{table*}

Given that few studies simultaneously consider impersonation and eavesdropping attacks in mmWave systems, we broaden our scope to include other wireless communication systems like non-orthogonal multiple access systems \cite{DBLP:journals/jsac/XieZCT22} and location service systems \cite{DBLP:journals/ton/LiZCCXL24}.  For instance, Xie \textit{et. al} utilize the channel response to generate the secret keys for message encryption and a hybrid authentication scheme is also designed by using the tag and channel responses to achieve identity validation \cite{DBLP:journals/jsac/XieZCT22}. On this basis, Li \textit{et. al} consider a cooperative attack scenario and propose an effective privacy-preserving physical layer authentication scheme to simultaneously protect privacy data and identity legitimacy \cite{DBLP:journals/ton/LiZCCXL24}. 
Nevertheless, the characteristics of the mmWave systems such as sparse propagation environment and rapidly varying channels degrade the feature distinguishability of the adopted channel response fingerprint. Also, these schemes only provide content protection without fundamentally mitigating eavesdropping threats and thus cannot be directly deployed in the mmWave BA stage. To outline the difference between this work and main related works \cite{DBLP:conf/infocom/WangJWLZ20}, \cite{DBLP:journals/tifs/TengZCJX24}, \cite{DBLP:journals/tifs/ZhangLZJX22}, \cite{DBLP:journals/twc/WangLN22}, \cite{DBLP:journals/jsac/XieZCT22}, we summarize in Table \ref{comparison} a comparison among these works in terms of system model, performance metrics, security guarantee, and main contributions.
\vspace{-1.2em}
\subsection{Limitations of Prior Art}
Although significant progress has been made in the design of PLS for mmWave systems by the works in \cite{DBLP:conf/infocom/WangJWLZ20}, \cite{DBLP:journals/tifs/NosouhiSGD22}, \cite{DBLP:journals/tifs/LuLSL23}, \cite{DBLP:journals/tifs/TengZCJX24},  \cite{DBLP:journals/tifs/ZhangLZJX22},  \cite{DBLP:journals/twc/JamaliM22}, \cite{DBLP:journals/twc/WangLN22}, \cite{DBLP:journals/twc/XiaoHLWSWY24}, several critical limitations remain unresolved.
First, these schemes tend to focus solely on either identity authentication or privacy protection, lacking a unified security mechanism that can simultaneously address both aspects in mmWave systems. 
Also, they mainly focus on established mmWave communication links, without considering the identity validation and covertness requirement during the BA phase. 
Second, we observe that the beam pattern is generally impacted by hardware impairments in antenna arrays (e.g., mutual coupling (MC) effects) \cite{DBLP:journals/tap/Schmid13}. It remains unclear whether such fine-grained imperfection could be leveraged to enhance authentication performance in mmWave systems, suggesting an opportunity for further research in incorporating MC effects into the secure authentication scheme design. Also, the current methods assign fixed weights to the adopted identity features and thus fail to dynamically adjust feature weights under different scenarios for improved authentication performance.
Finally, current security schemes primarily conduct theoretical analysis of only authentication or only covert communication performance, making it difficult to conduct a comprehensive performance analysis and to establish predictable guarantees. 
\vspace{-1.2em}
\subsection{Technical Challenges and Our Solutions}
Several technical challenges arouse when addressing the limitations of existing works. To address the limitation that existing security schemes consider only a single type of security concerns, we design a new security scheme tailored for the critical BA phase to simultaneously combat both the impersonation and eavesdropping attacks. The main challenge here lies in how to determine the proper signal transmission power and beam signal budget to meet the two conflicting security requirements in terms of authentication accuracy and covert transmission rate. To this end, we formulate the settings of these two parameters as a non-convex optimization problem and devise a corresponding iterative algorithm to solve this problem.
To address the limitation that existing authentication schemes lack of an adaptive weight setting for features, we design an adaptive authentication mechanism for the efficient device identity validation in the BA phase. The challenge therein is how to determine the optimal weight settings for all the received beam training signals so as to achieve the best authentication accuracy. To this end, we formulate the optimal weight setting as an optimization problem, and apply the sequential quadratic programming algorithm to tackle this highly complex and non-convex optimization problem.
To address the limitation that available theoretical frameworks are suitable for the analysis of only authentication performance or only covert communication performance, we develop a complete theoretical framework for the performance analysis and optimization of the proposed CovertAuth scheme, which involves the joint optimization of signal transmission power and beam signal budget, feature weight optimization, covert communication performance modeling, and authentication performance modeling.
\vspace{-1em}

\subsection{Contributions and Results}
Based on the above observations, we propose CovertAuth, a novel PLS framework suitable for the mmWave BA stage that simultaneously combats eavesdropping and identity-based impersonation attacks. The main contributions of this paper are as follows:
\vspace{-0em}
\begin{itemize}
    \item \textbf{Design of Covert Communication for the BA Phase.} Based on a predefined beam codebook, we first develop theoretical models for the BA performance metric and covert transmission rate. Under the imperfect channel state information (CSI), the covert communication problem is formulated to design the beam training budget and transmission power for maximizing the covert communication rate, subject to the required covertness constraint. The alternating optimization algorithm with successive convex approximation is adopted to iteratively find the optimal solution.
    \item  \textbf{MC Effect-based Authentication Scheme.} We first analyze the feature feasibility of the MC effect and explore its impact on the beam pattern. Considering the omnidirectional scanning in the BA phase, we assign different weight values to the received BA signals (involving the beam pattern impacted by the MC feature) and then design an adaptive weight-based physical layer authentication scheme.
    \item \textbf{Theoretical Analysis and Optimal Weight Determination.} We also provide theoretical modeling for the authentication performance metrics like detection and false alarm probabilities. To enhance the authentication performance of CovertAuth, a quantitative relationship is established between weight and performance metrics using a sum-weighted approximation technique. With this relationship, we formulate an optimization problem to determine the optimal weight allocation that maximizes authentication accuracy under a false alarm constraint.
    \item \textbf{Performance Validation.} Extensive simulations are conducted to validate the correctness of the theoretical models and evaluate the robustness of CovertAuth against impersonation and eavesdropping attacks. This study demonstrates the potential of the new CovertAuth scheme to simultaneously combat both identity impersonation and eavesdropping attacks, offering an attractive security solution for mmWave BA stage. Also, the proposed theoretical framework helps to reveal the inherent interplay between the covert communication requirement and authentication performance, and enables a flexible trade-off between such two security concerns to be initialized to adapt to varying application scenarios.
\end{itemize}

\vspace{-0em}

The remainder of this article is organized as follows. Section \ref{system model} illustrates the problem formulation and system model. The overview of CovertAuth is introduced in Section \ref{design_covertauth} and Section \ref{theoretical_analysis} presents the performance modeling and optimization analysis of CovertAuth. Numerical results are presented in Section \ref{num results} and Section \ref{Discussions} further discusses the effects of side-lobe on CovertAuth. Section \ref{related_work} reviews the related work, and concluding remarks are provided in Section \ref{conclusion}.

\vspace{-1em}
\section{Problem Formulation and System Model}\label{system model}
\begin{figure}
    \centering
    \includegraphics[width=0.9\linewidth]{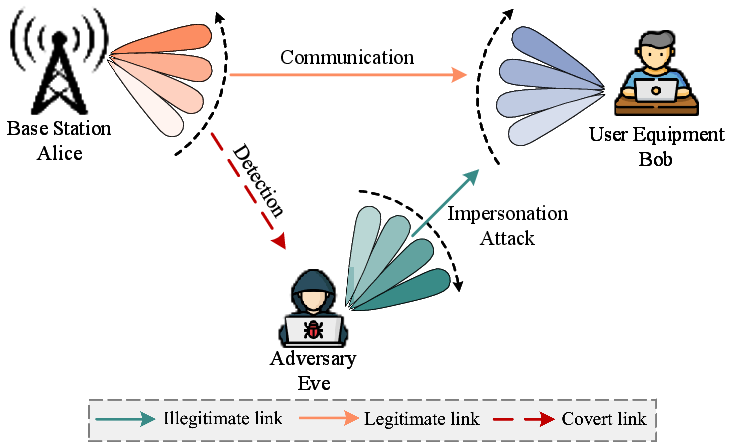}
    \setlength{\abovecaptionskip}{-0.3em}
    \caption{System model.}
    \label{fig_system}
    \vspace{-1.6em}
\end{figure}
\subsection{Problem Formulation}
As shown in Fig.~\ref{fig_system}, we consider a typical three-entity security model consisting of one legal base station Alice with $N_t$ antennas, a user equipment Bob with $N_r$ antennas, and a malicious adversary Eve with  $N_t$ antennas. The uniform linear array (ULA) is allocated with half-wavelength antenna spacing. 
To establish the mmWave downlink communication between Alice and Bob, Alice initially transmits reference signals with different transmit antenna patterns to identify the useful spatial directions in the channel environment. Then, Bob sweeps the beam codebook with a variety of receiving beams to search for the optimal beam pair with the maximal received power and then reports the selected beam pair to Alice to accomplish the beam training \cite{DBLP:journals/twc/YangZTS24}.  At the same time, the adversary Eve attempts to launch both passive and active attacks during the BA stage for confidential information recording and unauthorized access. The considered system model, although simple, can find applications in various 5G NR or IEEE 802.11ad-based systems, like vehicle-to-everything (V2X) and unmanned aerial vehicle (UAV) networks \cite{DBLP:journals/twc/ZhengWLJWWY22}, \cite{DBLP:journals/twc/HuangHYS25}. These scenarios generally involve critical security requirements for sensitive data protection and identity verification.

\vspace{-0.8em}
\subsection{Threat Model}
To evaluate the resilience of the proposed defensive framework CovertAuth, a powerful adversary is assumed to perform passive surveillance or active manipulation during the mmWave communication link establishment. For the passive surveillance operation, the adversary attempts to detect the location of the transmitter and wiretap the communication channel for the sensitive information. For active attacks, the adversary intends to cheat the receiver to pass the authentication by impersonating the identity of the legal transmitter.    
The prototypical offensive tactics that an adversary may exploit in the BA stage are as follows.
\begin{itemize}
    \item \textbf{Eavesdropping Attack.} In such a passive attack, Eve attempts to detect the position of the legal transmitter and stealthily record the transmitting signals. Then, She analyzes the captured signals for confidential information such as optimal beam pair selection and beampattern direction.
    \item \textbf{Identity-based Impersonation Attack.} We assume that Eve has prior knowledge about the adopted authentication scheme and communication protocol. Then, the adversary employs a device with the same model and hardware configuration to communicate with Bob by imitating the identity of Alice (e.g., the media access control address) for unauthorized access.
\end{itemize}
\vspace{-0.6em}
\subsection{Mutual Coupling Effect Model}
\textbf{Mathematical Model of MC effects.} The mutual coupling (MC) effects, as a kind of hardware imperfection in the antenna array, denote the electromagnetic interaction between antenna elements and are significantly influenced by the physical arrangement, antenna polarization, and manufacturing materials. For the considered ULA array, a symmetric Toeplitz matrix is generally exploited to model the MC structure \cite{DBLP:journals/tsp/AubryMLR23}.  If we use $M$ to denote the number of non-zero MC coefficients, then the MC matrix of the transmit array $\mathbf{C}_t\in\mathbb{C}^{N_t\times N_t}$ is modeled by
\begin{align}\label{MC_matrix}
    \mathbf{C}_t=\rm{Toeplitz}\{\mathbf{c}_t,0,\cdots,0\}.
\end{align}

The MC vector $\mathbf{c}_t=[1,c_{t,1},\cdots,c_{t,M-1}]^{T}\in\mathbb{C}^{M}$ and $c_{t,m}$ is the $m$-th non-zero element with $m=1,2,\cdots,M-1$. Following a similar way, we can obtain the MC matrix corresponding to the receiving antenna array $\mathbf{C_r}^{N_r\times N_r}$. Similar to previous work in \cite{DBLP:journals/tsp/ChenCCW19}, a complex Gaussian distribution with mean vector $\bar{\mathbf{c}}_t=[\bar{c}_{t,0},\bar{c}_{t,1},\cdots,\bar{c}_{t,M-1}]^T$  and covariance matrix $\sigma^2_{c}\mathbf{I}$ is adopted to generate the MC coefficient vector $\mathbf{c}_{t}$, that is, $\mathbf{c}_{t}\thicksim\mathcal{CN}(\bar{\mathbf{c}}_{t},\sigma^2_{c}\mathbf{I})$. In specific, $\bar{c}_{t,m}=\frac{\varsigma}{d_m}$ and $\sigma^2_{c}=\varsigma$, where $d_m$ denotes the normalized antenna spacing distance between the first and the $m$-th element and $\varsigma$ represents the combined effects of antenna material properties and manufacturing tolerance, respectively.

\textbf{Feature Feasibility Analysis of MC Effects.} We analyze the feasibility of the MC effects in terms of feature uniqueness and stability. On one hand, the feature uniqueness of MC is attributed to the fact that several factors like array geometry, array element tolerance, and manufacturing material properties collectively result in a unique coupling pattern, which is challenging to replicate across different antenna arrays. On the other hand, the feature stability of MC is rooted in the fixed physical structure and high-performance materials (e.g., high temperature and wear resistance) of the antenna array. The former minimizes mechanical deformation, while the latter preserves structural integrity and mitigates the impact of temperature variations on MC effects. The above characteristics present the authentication potential of MC as the hardware fingerprint.
\begin{figure}[htbp]
\vspace{-1em}
		\centering
        \hspace{-0.8em}
        \subfigure[Device 1 at different time.]
        {\includegraphics[width=0.51\linewidth]{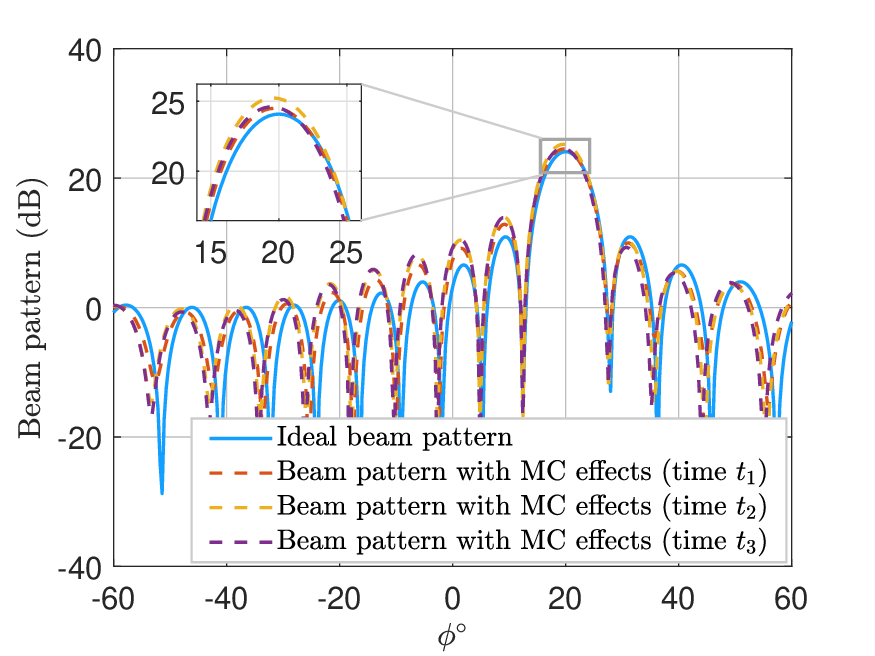}
        \label{MC_1}}
        \hspace{-1.6em}
        \subfigure[Device 2 at different time.]
        {\includegraphics[width=0.51\linewidth]{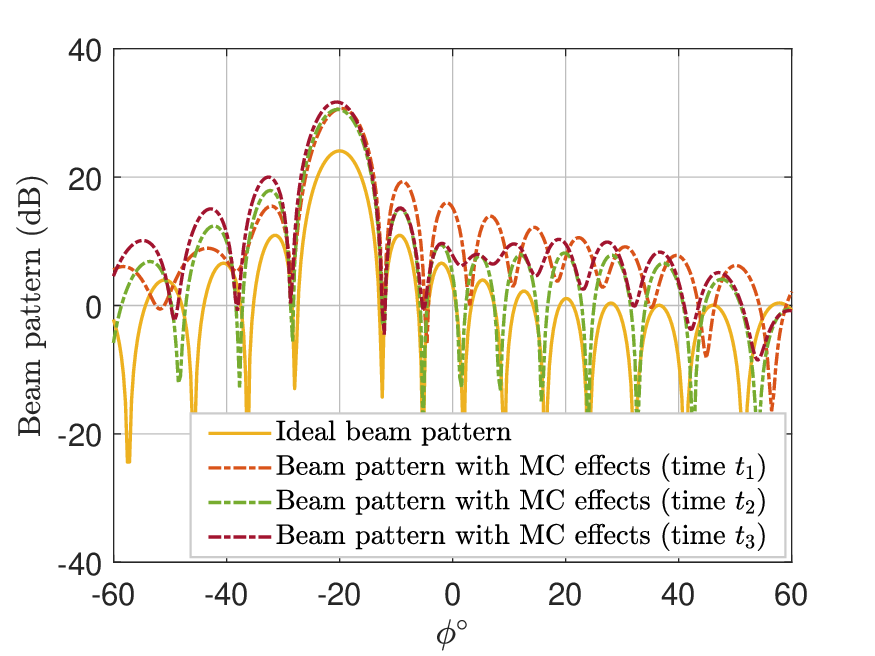}
        \label{MC_2}}
        \caption{Impacts of MC on transmit beam pattern using two devices with different antenna arrays.}
        \label{MC_effects}
\end{figure}

\textbf{Effects of MC on Transmit Beam Pattern.} The transmit beam pattern represents the radiation characteristics of an antenna array, depicting how signal power is distributed across different directions in space. For an ideal beam pattern with no hardware impairment, the response $B_{\rm{ideal}}$ is written as 
\begin{align}
    B_{\rm{ideal}}(\phi)= \mathbf{w}^{H}\mathbf{a}_t(\phi),
\end{align}
where $\mathbf{w}\in\mathbb{C}^{N_t}$ denotes the weight vector, representing the amplitude and phase adjustments applied to each antenna element. $\mathbf{a}_t(\phi)$  is the transmit array manifold vector, representing the phase response of the transmit array at the angle $\phi$.
In practice, however, hardware imperfections such as MC effects between antenna elements alter the transmit beam pattern. Then, the actual beam pattern becomes:
\begin{align}
    B_{\rm{MC}}(\phi)= \mathbf{w}^{H}\mathbf{C}_t\mathbf{a}_t(\phi).
\end{align}

To further present the feasibility of MC as a unique antenna array fingerprint, we investigate its impact on the array radiation pattern.
In this regard, we employ two transmitting devices with different antenna arrays and observe how the MC from two different antenna arrays affects the beam pattern under different time instants. The comparison results of the ideal and actual beam patterns are depicted in Fig.~\ref{MC_effects} and several interesting observations deserve attention.
First, we can see from Fig.~\ref{MC_1} and Fig.~\ref{MC_2} that the MC effect makes the actual beam pattern deviate from the ideal one, primarily manifested as increased side-lobe levels and slightly broadened main-lobe width.
Second, different devices with different antenna arrays lead to varying degrees of distortion in the beam pattern. This implies that the actual beam patterns with MC effects exhibit discernible device-specific characteristics. Moreover, due to the distinct locations of the two devices, the orientation angles of the main lobes are also markedly different, further enhancing the distinguishability of the beam pattern.
Third, one can note from Fig.~\ref{MC_1} or Fig.~\ref{MC_2} that a high degree of similarity between the beam patterns generated by the same device at different time instants, indicating the relatively stable nature of MC effects within a certain time frame.
The above observations motivate us to exploit the beam pattern feature incorporating MC effects to achieve identity authentication for the BA stage. 
\vspace{-0.5em}
\subsection{Signaling Model of Beam Alignment Stage}
Due to the sparse scattering environment in the mmWave communication system, a single-path line of sight (LoS) channel between Alice and Bob $\mathbf{H}\in\mathbb{C}^{N_r\times N_t}$ is considered here. If we use $\alpha$, $\theta$ and $\phi$ to denote the channel gain, angle of arrival (AoA), and angle of departure (AoD) respectively, then $\mathbf{H}$ with MC effects is characterized by 
\begin{align}
    \mathbf{H}= \alpha[\mathbf{C_r}\mathbf{a}_r(\theta)][\mathbf{C_t}\mathbf{a}_t(\phi)]^{H},              
\end{align}
where $\mathbf{a}_r(\theta)$ and $\mathbf{a}_t(\phi)$ denote the steering vectors associated with the AoA and AoD, respectively. They can be further written as 
\begin{align}
     \mathbf{a}_{r}(\theta)&=[1,e^{j\pi\sin(\theta)},\dots,e^{j\pi(N_{r}-1)\sin(\theta)}]^{T},\\
     \mathbf{a}_{t}(\phi)&=[1,e^{j\pi\sin(\phi)},\dots,e^{j\pi(N_{t}-1)\sin(\phi)}]^{T}.
\end{align} 

During the BA stage, an exhaustive search-based beam scanning scheme is adopted by Alice and Bob with the predefined beam codebooks $\mathcal{C}_{T}=\{\mathbf{w}_{l_T}\in\mathbb{C}^{N_t\times 1},l_{T}\in[1,2,\dots,L_T]\}$ and $\mathcal{C}_{R}=\{\mathbf{f}_{l_R}\in\mathbb{C}^{N_r\times 1},l_{R}\in[1,2,\dots,L_R]\}$. In specific, $\mathbf{w}_{l_T}\in\mathbb{C}^{N_t\times 1}$ and $\mathbf{f}_{l_R}\in\mathbb{C}^{N_r\times 1}$ denote the selected codewords from the codebooks by Alice and Bob, respectively. $L_T$ and $L_R$ are the size of $\mathcal{C}_{T}$ and $\mathcal{C}_{R}$, respectively. The $L_T$ unit-norm beams jointly cover the entire AoD region $\Phi$ and the $L_R$ unit-norm beams jointly span the entire AoA region $\Theta$. Then, the entire transceiver beam pair codebook $\mathcal{C}$ is composed of the Cartesian product of $\mathcal{C}_{T}$ and $\mathcal{C}_{R}$ (i.e., $\mathcal{C}=\{(\mathbf{w},\mathbf{f}):\mathbf{w}\in\mathcal{C}_T,\mathbf{f\in\mathcal{C}_R}\}$) with the size of $L=L_R L_T$. For ease of expression, we represent ($\mathbf{w}_l$, $\mathbf{f}_l$) as the $l$-th beam pair in the codebook $\mathcal{C}$, $l\in[1,L]$.

Based on the above beam scanning scheme, Alice utilizes the selected codeword $\mathbf{w}_l$ to transmit the pilot sequence $\mathbf{x}\in\mathbb{C}^{N}$ with beam training budget $N$ symbols and Bob employs the receiver beam $\mathbf{f}_l$ to measure the power of the received signal.  For the $l$-th beam pair, the output signal observed at the receiver side $\mathbf{y}\in\mathbb{C}^{N\times1}$ is formulated by
\begin{align}
    \mathbf{y}_l = \sqrt{P}\mathbf{f}_l^{H}\mathbf{H}\mathbf{w}_l\mathbf{x}+\mathbf{n},
\end{align}
where $P$ is the transmission power and the pilot sequence carries energy $||\mathbf{x}||^2_2=N$. Moreover, $\mathbf{n}\in\mathbb{C}^{N\times 1}$ denotes the zero-mean complex Gaussian noise with variance $\sigma^2_n$. 

After receiving $\mathbf{y}_l$, Bob applies the known $\mathbf{x}$ to conduct a match-filtered output:
\begin{align}
    y_l = \mathbf{x}^H\mathbf{y}_l=\sqrt{P}N\mathbf{f}_l^{H}\mathbf{H}\mathbf{w}_l +\Tilde{n},\ \ \ \ l\in[1,L],
\end{align}
where $\Tilde{n}=\mathbf{x}^H\mathbf{n}\thicksim\mathcal{CN}(0,N\sigma^2_n)$. Then, SNR of $y_l$ is defined as SNR = $\frac{NP|\mathbf{f}_l^{H}\mathbf{H}\mathbf{w}_l|^2}{\sigma^2_n}$. For the received $L$ beam signals, Bob selects the beam pair that maximizes the match-filtered output as the beam establishment candidate:
\begin{align}
    \hat{l}=\underset{l \in[1: L]}{\arg \max }\left|y_l\right|.
\end{align}

Following the BA stage, the selected beam pair will be used for data transmission. We assume that there are $N_{\rm{total}}$ symbols in a time slot, among them $N$ symbols are used for beam alignment. Similar to \cite{DBLP:journals/tcom/ZhangHZY20}, we model the pilot overhead as a discount in effective rate and denote it as $1-\frac{N}{N_{\rm{total}}}$. We use $R$ to denote the communication rate between Alice and Bob during the data transmission stage. Given the effects of BA quality on $R$, we adopt the concept of average effective rate to characterize $R$ as
\begin{align}\label{average_capcity}
    R=\left(1-\frac{N}{N_{\rm{total}}}\right)\mathbb{E}\left\{\log\left(1+\frac{NP|\mathbf{f}_{\hat{l}}^{H}\mathbf{H}\mathbf{w}_{\hat{l}}|^2}{\sigma^2_n}\right)\right\}.
\end{align}
\vspace{-2em}
\section{CovertAuth Scheme}\label{design_covertauth}
In this section, we first present an overview of the proposed CovertAuth scheme, consisting of signal optimization-based covert communication mechanism and adaptive weight-based authentication mechanism. Then, we introduce the detailed processes involved in both mechanisms to demonstrate how CovertAuth combats both eavesdropping and impersonation attacks. 

\vspace{-0.8em}
\subsection{Covert Communication Mechanism Design}
The covert communication mechanism comprises four steps: theoretical modeling of covert communication rate and covertness constraint, optimization problem formulation, iterative algorithm solution, and BA operation. In specific, CovertAuth first exploits the knowledge of order statistics to derive the closed-form expression of the successful BA probability as the performance metric to evaluate the BA quality. By considering the impact of successful BA probability on communication quality, the theoretical model of the covert communication rate is also developed. Then, the covertness constraint on the adversary is characterized by using the Kullback-Leibler (KL) divergence technique and the metric of detection error probability. According to the covertness requirement of the communication system, we design a secure communication mechanism by formulating the settings of transmission power $P$ and beam training budget $N$ as an optimization problem with the goal of maximizing the covert communication rate between authorized entities. Due to the non-convex nature of the problem, an alternating optimization iteration algorithm is used to solve it.
Finally, the transmitter leverages the designed beam signals to perform beam alignment with the receiver. 
It can be seen that the above mechanism carefully designs transmission parameters ($P$ and $N$) to ensure that the communication behavior cannot be detected by malicious attackers and thus to effectively combat eavesdropping attacks.

\subsection{Adaptive Weight-based Authentication Mechanism Design}
The adaptive weight-based authentication mechanism consists of five steps: theoretical modeling of authentication performance metrics, revealing the relationship between weights and beamforming gain, optimization problem formulation, problem solution, and identity authentication. It is noted that the received BA signals $y_l, l\in[1,L]$ contain the beam pattern feature stemming from both the device-specific MC effect and spatial information of the transmitter. Therefore, CovertAuth performs a weighted summation of $L$ beam pair signals to design an adaptive weight-based identity discriminator. Before making an identity decision with such a discriminator, it is significant to determine the settings of the concerned feature weights. To this end, we first use the statistical signal processing technique to develop theoretical models of successful detection and false alarm probabilities. Based on such models, we reveal the functional relationship between feature weights and beamforming gains of the received $L$ beam pair signals. Then, given a constrained false alarm probability, the determination of the feature weights is formulated as an optimization problem and such a highly complex problem is solved by a sequential quadratic programming algorithm. Finally, based on the obtained adaptive weights, the receiver exploits the above-mentioned weighted energy detector to achieve an identity validation of the transmitter and thus effectively combat the impersonation attacks.

\vspace{-2em}
\section{Performance Modeling and Optimization}\label{theoretical_analysis}
In this section, we first present the covert communication mechanism design of CovertAuth. Then, an adaptive feature weight-based authentication mechanism and the determination of the concerned feature weights are also introduced. Finally, we conduct a convergence analysis and computational complexity of the CovertAuth scheme. 
\vspace{-1.4em}
\subsection{Covert Communication Mechanism Design}\label{covert_phase}
\textbf{Theoretical Modeling of Communication Rate.} 
To derive the closed-form expression of $R$ in \eqref{average_capcity}, we proceed in two steps: 1) accurately characterize BA performance, and 2) derive an analytical expression for the average rate model that captures the effects of BA performance. 
For the first step, we intend to quantify the effective channel gain corresponding to the $l$-th beam pair $g_l=|\mathbf{f}_l^{H}\mathbf{H}\mathbf{w}_l|^2$ with a predefined beam codebook and then exploit the probability of successful beam alignment $P_{\rm{a}}$ to evaluate BA quality. In particular, $g_l$ is further written as
\begin{align}
    g_l= &|\mathbf{f}_l^{H}\mathbf{H}\mathbf{w}_l|^2\nonumber\\
    =&|\alpha|^2|\mathbf{f}_l^{H}\Tilde{\mathbf{a}}_r(\theta)\Tilde{\mathbf{a}}_t(\phi)^H\mathbf{w}_l|^2\nonumber\\
    =&|\alpha|^2F_l(\theta)W_l(\phi),
\end{align}
where $\Tilde{\mathbf{a}}_t(\phi)=\mathbf{C_t}\mathbf{a}_t(\phi)$ and $\Tilde{\mathbf{a}}_r(\theta)=\mathbf{C_r}\mathbf{a}_r(\theta)$ denote the actual steering vectors at the transceiver ends impacted by the MC effects. In addition, $W_l(\phi)=|\Tilde{\mathbf{a}}_t(\phi)^H\mathbf{w}_l|^2$ and $F_l(\theta)=|\mathbf{f}_l^{H}\Tilde{\mathbf{a}}_r(\theta)|^2$ are beamforming gain at the transmitter and receiver side, respectively. Then, we follow the operation in \cite{DBLP:journals/tcom/ZhangHSWY17} to quantify $W_l(\phi)$ and $F_l(\theta)$:
\begin{align}
W_l(\phi) & = \begin{cases}W_T \triangleq \frac{4 \pi}{\left|\Omega_T\right| / L_T}, & \text { if } \phi \in \Phi_{\mathbf{w}_l}, \\
0, & \text { otherwise},\end{cases} \\
F_l(\theta) & = \begin{cases}F_R \triangleq \frac{4 \pi}{\left|\Omega_R\right| / L_R}, & \text { if } \theta \in \Theta_{\mathbf{f}_l}, \\
0, & \text { otherwise},\end{cases}
\end{align}
where $\Omega_T$ and $\Omega_R$ denote the solid angles covering the entire AoD and AoA region. Moreover, $\Phi_{\mathbf{w}_l}$ and $\Theta_{\mathbf{f}_l}$ are the covered regions related to the beam pair $\mathbf{w}_l$ and $\mathbf{f}_l$, respectively. The above quantization result implies that when $\phi$ and $\theta$ fall within the coverage area $\Phi_{\mathbf{w}_l}$ and $\Theta_{\mathbf{f}_l}$, respectively, the maximum transmitting and receiving beamforming gain can be obtained as $W_T$ and $F_R$. Otherwise, the beamforming gain is quantified as zero. Accordingly, $g_l$ is quantified as
\begin{align}\label{quantified_channel_gain}
    g_l & = \begin{cases}|\alpha|^2F_RW_T, & \text { if } \phi \in \Phi_{\mathbf{w}_l}\ \text{and}\  \theta \in \Theta_{\mathbf{f}_l} , \\
0, & \text { otherwise}.\end{cases} 
\end{align}

It is evident that the precision of BA significantly affects the beamforming gain, and thus the metric for measuring BA performance is of particular interest. We propose to adopt the BA successful probability $P_{\rm{a}}$ as the evaluation metric. Without loss of generality, it is assumed that the optimal beam pair is $l=1$. To obtain the analytical expression of $P_{\rm{a}}$, we define $Y_l=\frac{2|y_l|^2}{N\sigma^2_n}$ and explore the statistic information of $Y_l$.
With the quantified channel gain $g_l$ in \eqref{quantified_channel_gain}, $y_l$ is a complex Gaussian variable with mean $N\sqrt{Pg_l}$ and variance $N\sigma^2_n$. Then, the normalized statistic $Y_l$ obeys a non-central Chi-square distribution $\chi^2_2(\lambda_l)$ with non-central parameter $\lambda_l=\frac{2NPg_l}{\sigma^2_n}$ and degrees of freedom (DoF) 2. In specific, for $l=1$, $Y_1\thicksim\chi^2_2(\lambda_1)$ with $\lambda_1=\frac{2|\alpha|^2NPF_RW_T}{\sigma^2_n}$; otherwise, $Y_l\thicksim\chi^2_2(0),l\in[2,L]$. The successful beam alignment event occurs when $\hat{l}=1$, thus $P_{\rm{a}}$ is mathematically expressed by 
\begin{align}
    P_{\rm{a}}&=1-\mathbf{Pr}(\hat{l}\neq1)\nonumber\\
    &=1-\mathbf{Pr}(Y_1<\max\{Y_2,\cdots,Y_L\}).
\end{align}
Then, we rank $(Y_2, Y_3, \cdots Y_L)$ in an ascending order ($Y_1^{'}, Y_2^{'}, \cdots, Y_{L-1}^{'}$) and $P_{\rm{a}}$ is rewritten as
\begin{align}
    P_{\rm{a}} = 1-\mathbf{Pr}(Y_1<Y_{L-1}^{'}).
\end{align}

We present the following Theorem \ref{theorem_1} to give the closed-form expression of $P_{\rm{a}}$:
\begin{theorem}\label{theorem_1}
    For the adopted exhaustive search-based beam training scheme, let $P_{\rm{a}}$ denote the probability of successful beam alignment. Then, the closed-form expression of $P_{\rm{a}}$ is derived as follows:
    \begin{align}\label{pa}
        P_{\rm{a}}&=1-\sum^{L-2}_{l=0}\left[(-1)^l\frac{\binom{L-2}{l}}{(l+1)(l+2)}\exp\left(-\frac{\lambda_1(l+1)}{2(l+2)}\right)\right]\nonumber\\
        &\ \ \ \times(L-1), 
    \end{align}
    where $\binom{L-2}{l}=\frac{(L-2)!}{l!(L-2-l)!}$ is a combinatorial coefficient.
\end{theorem}
\begin{proof}
    The results can be obtained by following a similar derivation process in \cite{DBLP:journals/tifs/ZhangLYLCZW21}.
\end{proof}
Notice that the effective channel gain $g_{\hat{l}}$ is determined by the beam alignment performance. In particular, if the successful beam alignment event occurs, $g_{\hat{l}}$ is quantified as $\left|\alpha\right|^2F_RW_T$ with the probability $P_a$; Otherwise, $g_{\hat{l}}$ is quantified as zero with the probability $1-P_a$. By taking the expectation with respect to $g_{\hat{l}}$, we can derive the analytical expression for the average communication rate $R$ as
\begin{align}\label{covert_rate_ideal}
    R&=\left(1-\frac{N}{N_{\rm{total}}}\right)\mathbb{E}_{g_{\hat{l}}}\left\{\log\left(1+\frac{NPg_{\hat{l}}}{\sigma^2_n}\right)\right\}.\nonumber\\
    &=\left(1-\frac{N}{N_{\rm{total}}}\right)P_{\rm{a}}\log\left(1+\frac{|\alpha|^2NPF_RW_T}{\sigma^2_n}\right).
\end{align}

\textbf{Characterization of Covertness Constraint.}
In this part, we adopt the detection performance at the Eve side to evaluate the covertness level during the BA phase \cite{DBLP:journals/tifs/ZhangLZJX22}. Accordingly, a binary hypothesis testing of the received $l$-th beam pair signal $y_{e,l}, l\in[1,L]$ is formulated as
\begin{equation}\label{Detection_hypothesis}
    \left\{\begin{array}{ll}
        \mathcal{D}_0:&y_{e,l}=\Tilde{v}_l,\\
        \mathcal{D}_1:&y_{e,l}=N\sqrt{P}\mathbf{h}_{e}\mathbf{w}_l+\Tilde{v}_l,
       \end{array}\right.
\end{equation}
where $\Tilde{v}_l$ is a zero-mean Complex Gaussian noise with variance $N\sigma^2_e$ at Eve and $\mathbf{h}_{e}=\alpha_e\Tilde{\bm{a}}_t(\phi)^H \in\mathbb{C}^{1\times N_t}$ denotes the channel between Alice and Eve with complex gain $\alpha_e$. In particular, $\mathcal{D}_0$ represents the null hypothesis, indicating that Alice keeps silent. Under $\mathcal{D}_0$, $y_{e,l}$ is independently and identically distributed following $y_{e,l}\thicksim\mathcal{CN}(0,N\sigma^2_{e})$.
Conversely, $\mathcal{D}_1$ is the alternative hypothesis, indicating that Alice performs a transmission behavior with Bob. In this case, $y_{e,l}\thicksim\mathcal{CN}(N\sqrt{P}\mathbf{h}_{e}\mathbf{w}_l,N\sigma^2_{e})$. After stacking all $L$ beam pair signals into a column, we can obtain PDFs of $\mathbf{y}_{e}=[y_{e,1},y_{e,2},\cdots,y_{e,L}]^T\in\mathbb{C}^{L\times1}$ under $\mathcal{D}_0$ and $\mathcal{D}_1$ (denoted by $P_0(\mathbf{y}_{e})$ and $P_1(\mathbf{y}_{e})$, respectively):
\begin{align}
    P_0(\mathbf{y}_{e})&=\frac{1}{[\pi(N\sigma^2_e)]^L}\exp\left(-\frac{\sum^L_{l=1}|y_{e,l}|^2}{N\sigma^2_e}\right),\\
    P_1(\mathbf{y}_{e})&=\frac{1}{[\pi(N\sigma^2_e)]^L}\exp\left(-\frac{\sum^L_{l=1}|y_{e,l}-N\sqrt{P}\mathbf{h}_{e}\mathbf{w}_l|^2}{N\sigma^2_e}\right).
\end{align}
We employ the KL divergence $\mathcal{D}(P_0(\mathbf{y}_{e})||P_1(\mathbf{y}_{e}))$ to characterize the lower bound of minimum total detection error probability $\xi^{*}$. Let $\epsilon$ denote the required covertness level, the covertness constraint imposed on Eve can be formulated by using Pinsker's inequality as 
\begin{align}\label{covert_constraint}
    \frac{\sum_{l=1}^{L}NP|\mathbf{h}_e\mathbf{w}_l|^2}{\sigma^2_e}\leq2\epsilon^2.
\end{align}

\textbf{Design of $N$ and $P$ for Covert Communication.} To prevent eavesdropping attacks, we aim to 
maximize the covert rate between Alice and Bob by carefully designing the beam training budget $N$ and transmitting power $P$, subject to the covertness constraint. Due to the channel estimation error and non-cooperation between Alice and Eve, the imperfect channel with a bounded error model is also assumed:
\begin{align}\label{imperfect_CSI_constraint}
    \mathbf{h}_e=\hat{\mathbf{h}}_e+\Delta \mathbf{h}_e,\ \ \ \forall||\Delta\mathbf{h}_e||^2\leq h_e^2,
\end{align}
where $\hat{\mathbf{h}}_e$ and $\Delta \mathbf{h}_e$ denote the imperfect channel and channel estimation error, respectively. The norm of estimation error $\Delta \mathbf{h}_e$ is bounded by a known constant $h_e^2$. Let $P_{\text {max }}$ and $N_{\text {max }}$ denote the maximum power on the transmitter and beam training budget, respectively.
Therefore, under the imperfect channel state information scenario, the optimization problem with respect to $N$ and $P$ is formulated as
\begin{subequations}\label{problem_opt}
\begin{align}
\max _{P, N}& \ \left(1-\frac{N}{N_{\rm{total}}}\right)P_{\rm{a}}\log\left(1+\frac{|\alpha|^2NPF_RW_T}{\sigma^2_n}\right),  \\ 
\text { s.t. }  &0<P \leq P_{\text {max }}\label{constraint_P}, \\
& 1 \leq N \leq  N_{\text {max }}\label{constraint_N}, \\
& \eqref{covert_constraint}, \eqref{imperfect_CSI_constraint}.\nonumber
\end{align}
\end{subequations}
It is noted that existing algorithms make it difficult to solve the above optimization problem due to the following reasons: 1) the coupling between $N$ and $P$ involved in constraint \eqref{covert_constraint} and objective function; 2) the uncertainty in channel estimation leads to the infinite inequality constraints in \eqref{imperfect_CSI_constraint}; and 3) the discrete variable $N$ results in a mixed integer non-convex optimization problem. 

To tackle these problems, we first exploit the inequality transformation to obtain a more tractable manner. By using \cite[Eq. 26]{DBLP:journals/tsp/ZhaoCSCZ16}, we can obtain the following inequality
\begin{align}
    \left|(\hat{\mathbf{h}}_e+\Delta\mathbf{h}_e)\mathbf{w}_l\right|\leq\left|\hat{\mathbf{h}}_e\mathbf{w}_l\right|+\left|\Delta\mathbf{h}_e\mathbf{w}_l\right|\leq\left|\hat{\mathbf{h}}_e\mathbf{w}_l\right|+h_e||\mathbf{w}_l||.
\end{align}
Based on this inequality transformation,  the constraints in \eqref{covert_constraint} and \eqref{imperfect_CSI_constraint} can be written into a more compact manner:
\begin{align}
    NP\sum_{l=1}^{L}\left||\hat{\mathbf{h}}_e\mathbf{w}_l|+h_e|\mathbf{w}_l|\right|^2\leq2\sigma^2_e\epsilon^2.
\end{align}

To address the couple of variables in the objective function and constraints, we adopt the Lagrange dual-decomposition method in \cite{DBLP:journals/tifs/ZhangLYLCZW21} to decouple the problem \eqref{problem_opt} into a master dual problem and a subproblem. Let $\nu$ denote the nonnegative multiplier and $\bm{\varsigma}=[N,P]^T$ denote the optimization variable vector, then the Lagrange function is given by
\begin{align}
    \mathcal{L}&(\nu;\bm{\varsigma})\nonumber\\
    &=\log\left(1-\frac{N}{N_{\rm{total}}}\right)+\log\left(\log\left(1+\frac{NP|\alpha|^2F_RW_T}{\sigma^2_n}\right)\right)\nonumber\\
    &\ \ \ +\log P_{\rm{a}}-\nu\left(NP\sum_{l=1}^{L}\left||\hat{\mathbf{h}}_e\mathbf{w}_l|+h_e|\mathbf{w}_l|\right|^2-2\sigma^2_n\epsilon^2\right).
\end{align}
Accordingly, the dual function $\mathcal{F}(\nu)$ is formulated by 
\begin{align}
     \mathcal{F}(\nu)&=\max_{\bm{\varsigma}}\mathcal{L}(\nu;\bm{\varsigma}).
\end{align}
For the given dual variable $\nu$, the subproblem in terms of $\bm{\varsigma}$ is given by
\begin{subequations}
    \begin{align}\label{sub_problem}
    \max_{\bm{\varsigma}}&\ \log\left(1-\frac{N}{N_{\rm{total}}}\right)+\log\left(\log\left(1+\frac{NP|\alpha|^2F_RW_T}{\sigma^2_n}\right)\right)\nonumber\\
    &+\log P_{\rm{a}}-\nu NP\Gamma,\\
    \text{s.t.}&\  \eqref{constraint_P}, \eqref{constraint_N},\nonumber
\end{align}
\end{subequations}
where $\Gamma=\sum_{l=1}^{L}\left||\hat{\mathbf{h}}_e\mathbf{w}_l|+h_e|\mathbf{w}_l|\right|^2$. Based on the optimal variable $\bm{\varsigma}$, the master dual problem with respect to $\nu$ is as follows:
\begin{subequations}
\begin{align}\label{master_problem}
    \min_{\nu}&\  \mathcal{F}(\nu),\\
    \text{s.t.} &\  \eqref{constraint_P}, \eqref{constraint_N}\nonumber.
\end{align}
\end{subequations}
As to the master dual problem in \eqref{master_problem} and the subproblem in \eqref{sub_problem}, an alternating optimization algorithm combined with the successive convex approximation (SCA) method is employed to achieve an optimal solution of $\nu$ and $\bm{\varsigma}$. Let the superscript $t$ be the $t$-th iteration. 
\subsubsection{Update $\nu$ for fixed $N$ and $P$}Since the master dual problem in \eqref{master_problem} is always convex, we propose to adopt the subgradient method to optimize $\nu^{t+1}$ given the fixed $\bm{\varsigma}^t$, that is,
\begin{align}\label{solver_nu}
    \nu^{t+1}=\max\left(0,\nu^t+\eta\left( N^tP^t\Gamma-2\sigma^2_e\epsilon^2\right)\right),
\end{align}
where $\eta$ denotes the step size and operator $\max(0,a)$ ensures that the result is non-negative.
\subsubsection{Update $P$ for fixed $N$ and $\nu$} For the given $N$ and $\nu$, the optimization subproblem with respect to $P$ can be formulated as
\begin{subequations}\label{sub_problem_P}
    \begin{align}
    \max_{P}&\ \log P_{\rm{a}}+\log\left(\log\left(1+\frac{NP|\alpha|^2F_RW_T}{\sigma^2_n}\right)\right)\nonumber\\
    &\ \ -\nu NP\Gamma,\\
    \text{s.t.}&\  \eqref{constraint_P}.\nonumber
\end{align}
\end{subequations}
It is obvious that the above problem is difficult to solve due to the non-concave part of the objective function. To this end, we first partition the objective function into concave and non-concave parts (denoted by $\mathcal{F}_{\rm{c}}(P)$ and $\mathcal{F}_{\rm{n}}(P)$, respectively):
\begin{align}
    \mathcal{F}_{\rm{c}}(P)=&\log\left(\log\left(1+\frac{NP|\alpha|^2F_RW_T}{\sigma^2_n}\right)\right)-\nu NP\Gamma,\\
    \mathcal{F}_{\rm{n}}(P)=&\log P_{\rm{a}} \label{P_a_P}.
\end{align}
Then, we retain the partial concavity of the original objective function $\mathcal{F}_{\rm{c}}(P)$ and linearize the non-concave part $\mathcal{F}_{\rm{n}}(P)$. This allows us to construct a concave surrogate function $\Tilde{\mathcal{F}}(P)$ that approximates the original function while maintaining the desirable properties for optimization:
\begin{align}\label{linear_P}
    \Tilde{\mathcal{F}}(P)=\mathcal{F}_{\rm{c}}(P)+\frac{\partial \mathcal{F}_{\rm{n}}(P)}{\partial P}\bigg|_{P=P^t}\left(P-P^t\right)-\tau_p\left(P-P^t\right)^2,
\end{align}
where $\tau_p$ is a positive constant such that the surrogate function $\Tilde{\mathcal{F}}(P)$ is strongly concave. $\frac{\partial \mathcal{F}_{\rm{n}}(P)}{\partial P}$ is the partial derivative of $\mathcal{F}_{\rm{n}}(P)$ with respect to $P$.
Finally, the optimal value of $P^{t+1}$ is obtained by solving the following concave optimization problem:
\begin{subequations}
  \begin{align}\label{solver_P}
    \max _{P}&\ \ \Tilde{\mathcal{F}}(P),\\
    \text{s.t.}&\  \eqref{constraint_P}.\nonumber
\end{align}
\end{subequations}
It is obvious that the above problem can be effectively solved by using the CVX toolbox \cite{boyd2004convex}.
\subsubsection{Update $N$ for fixed $P$ and $\nu$}With fixed $P$ and $\nu$, the optimization for $N$ is formulated as
\begin{subequations}
    \begin{align}
        \max_{N}&\ \log P_{\rm{a}}+\log\left(\log\left(1+\frac{NP|\alpha|^2F_RW_T}{\sigma^2_n}\right)\right)\nonumber\\
    &\ \ +\log\left(1-\frac{N}{N_{\rm{total}}}\right)-\nu NP\Gamma,\\
    \text{s.t.}&\  \eqref{constraint_N}.\nonumber
    \end{align}
\end{subequations}
Similar to the optimization process for $P$, we also partition the objective function into concave and non-concave parts (denoted by $\mathcal{F}_{\rm{c}}(N)$ and $\mathcal{F}_{\rm{n}}(N)$, respectively):
\begin{align}
    \mathcal{F}_{\rm{c}}(N)=&\log\left(\log\left(1+\frac{NP|\alpha|^2F_RW_T}{\sigma^2_n}\right)\right)\nonumber\\
    &+\log\left(1-\frac{N}{N_{\rm{total}}}\right)-\nu NP\Gamma,\\
    \mathcal{F}_{\rm{n}}(N)=&\log P_{\rm{a}}.\label{P_a_N}
\end{align}
The concave surrogate function $\Tilde{\mathcal{F}}(P)$ is written as
\begin{align}\label{linear_N}
    \Tilde{\mathcal{F}}(N)&=\mathcal{F}_c(N)+\frac{\partial \mathcal{F}_{\rm{n}}(N)}{\partial N}\bigg|_{N=N^t}\left(N-N^t\right)\nonumber\\
    &\ \ \ -\tau_N\left(N-N^t\right)^2,
\end{align}
where $\tau_N$ is a positive constant ensuring that $\Tilde{\mathcal{F}}(N)$ is strongly concave.  $\frac{\partial \mathcal{F}_{\rm{n}}(N)}{\partial N}$ is the derivative part of $\mathcal{F}_{\rm{n}}(N)$ with respect to $N$.
Finally, the variable $N$ is optimized by solving the following problem
\begin{subequations}
  \begin{align}\label{solver_N}
    \max _{N}&\ \ \Tilde{\mathcal{F}}(N),\\
    \text{s.t.}&\  \eqref{constraint_N}.\nonumber
\end{align}
\end{subequations}
Obviously, the above problem can be solved by using existing algorithms or the CVX toolbox. To ensure that the obtained $N^{t+1}$ conforms to the constraint of a discrete integer, we round $N^{t+1}$ to its nearest integer:
\begin{align}\label{round_N}
    N^{t+1}=\left\lfloor N^{t+1}+\frac{1}{2}\right\rfloor.
\end{align}

Based on the above process, the whole design procedure is summarized in Algorithm \ref{Algothrim}.
\begin{algorithm}
	\caption{Joint Transmitting Power and Beam Training Budget Design Algorithm.}
        \label{Algothrim}
        \renewcommand{\algorithmicrequire}{ \textbf{Input:}}     
        \renewcommand{\algorithmicensure}{ \textbf{Output:}}  
	\begin{algorithmic}[1]  
		\REQUIRE $\Gamma$, $\sigma^2_e$, $\sigma^2_n$, $t_{\rm{max}}$, tolerance $\Tilde{\rho}$, $\epsilon$, $\eta$, $\tau_p$, $\tau_N$.
        \ENSURE  $P^\star$ and $N^\star$.    
        \STATE \textbf{Initialization:}~ feasible initial point ($\nu^0$ ,$P^0$, $N^0$), $t=0$, $\rho^0=0$.
        \STATE \textbf{while} $\rho^t>\Tilde{\rho}$ and $t\leq t_{\rm{max}}$ \textbf{do}
        \STATE \ \ Update $\nu^{t+1}$ by \eqref{solver_nu}.
        \STATE \ \ Update $P^{t+1}$ by solving \eqref{solver_P}.
        \STATE \ \ Update $N^{t+1}$ by solving \eqref{solver_N}.
        \STATE \ \ Round $N^{t+1}$ by \eqref{round_N}.
        \STATE \ \ Calculate $\rho^{t+1}=\left|\mathcal{L}(\nu^{t+1};\bm{\varsigma}^{t+1})-\mathcal{L}(\nu^{t};\bm{\varsigma}^{t})\right|$.
        \STATE \ \ $t=t+1$.
        \STATE \textbf{end while}
        \RETURN $P^\star=P^t$ and $N^\star=N^t$.
	\end{algorithmic}
\end{algorithm}
\vspace{-0.8em}
\subsection{Adaptive weight-based Authentication Mechanism Design}\label{authentication_phase}
\textbf{Authentication Criterion with $N^\star$ and $P^\star$.} For the receiver Bob, it is imperative to authenticate the identity of the signals and then decide whether to establish a communication link with the transmitter. Due to the MC feature involved in the beam pair signals $y_l$, Bob intends to exploit $y_l$ to accomplish identity verification. The identity validation criterion can be formulated by using a binary hypothesis testing:
\begin{equation}\label{authentication_hypothesis}
    \left\{\begin{array}{ll}
\mathcal{H}_0:&y_{l}=N^\star_0\sqrt{P^\star_0}\mathbf{f}_l^{H}\alpha_0\Tilde{\mathbf{a}}_{r,0}(\theta_0)\Tilde{\mathbf{a}}_{t,0}^{H}(\phi_0)\mathbf{w}_l+\Tilde{n}_l,\\
\mathcal{H}_1:&y_{l}=N_1\sqrt{P_1}\mathbf{f}_l^{H}\alpha_1\Tilde{\mathbf{a}}_{r,1}(\theta_1)\Tilde{\mathbf{a}}_{t,1}^{H}(\phi_1)\mathbf{w}_l+\Tilde{n}_l.
       \end{array}\right.
\end{equation}
In specific, under the null hypothesis $\mathcal{H}_0$, $\Tilde{\mathbf{a}}_{r,0}(\theta_0)=\mathbf{C_{r}}_{,0}\mathbf{a}_r(\theta_0)$ and $\Tilde{\mathbf{a}}_{t,0}(\phi_0)=\mathbf{C_{t}}_{,0}\mathbf{a}_t(\phi_0)$ denote the steering vectors containing both MC feature and spatial information of the authorized transmitter, indicating that the signals are from Alice. In contrast, on $\mathcal{H}_1$, $\Tilde{\mathbf{a}}_{r,1}(\theta_1)=\mathbf{C_{r}}_{,0}\mathbf{a}_r(\theta_1)$ and $\Tilde{\mathbf{a}}_{t,1}(\phi_1)=\mathbf{C_{t}}_{,1}\mathbf{a}_t(\phi_1)$ denote the steering vectors containing both MC feature and spatial information of the unauthorized transmitter, implying that the current transmitter is Eve. 
Due to the different wireless channels (i.e., Alice-Bob and Eve-Bob), it is difficult for Eve to possess the same beam training budget $N$ and transmitting power $P$ as those related to Alice. Herein, we consider a challenging scenario (i.e., $N_1\approx N^\star_0$ and $P_1\approx P^\star_0$) to evaluate the robustness of the proposed authentication method. 

\textbf{Sum-weighted Authentication Decision.}
As the first attempt to integrate covert communication mechanism with authentication mechanism to simultaneously combat both eavesdropping and impersonation attacks in the mmWave BA stage, this paper considers a simple energy detector-based authentication decision mechanism widely adopted in available works \cite{DBLP:journals/tvt/HuangW18}, \cite{DBLP:journals/iotj/HeNZQ24}, \cite{DBLP:journals/tdsc/ZhangTNJF24}. Such a study helps us to have a quick view for the essential benefits and feasibility of integration of these two security mechanisms, and thus lays the foundation for future research incorporating more advanced authentication techniques. It is notable that the energy detector method, although simple, owns the advantages of low computational complexity and no prior knowledge requirement for adversary. These advantages are particular attractive for the mmWave systems concerned in this work. This is because the mmWave systems are usually resource-constrained and latency-sensitive, and the dynamic and complex mmWave propagation environment makes it highly difficult for the receiver to obtain prior statistical knowledge about the adversary (e.g., transmission power $N_1$ and directional angles $\theta_1$, $\phi_1$). In addition, note that each beam pair signal $y_l$ carries a certain degree of MC fingerprint and spatial location information of the transmitter. 
Thus, we assign various weights $\omega_l$ to the $L$ beam pair signals $y_l, l\in[1,L]$ and then exploit the energy detector to design the final identity discriminator:
\begin{align}\label{decision}
    T =\sum^{L}_{l=1}\omega_l\frac{2|y_l|^2}{N^\star\sigma^2_n}\mathop{\gtrless}_{H_0}^{H_1}\tau,
\end{align}
where $T$ is the decision test statistic and $\tau$ is the threshold. Based on this identity discriminator, Bob can effectively validate the identity legitimacy of the current transmitter.
\begin{remark}
    The optimality of the detector depends on the availability of the signal statistical knowledge.  It is shown in \cite{kay1993fundamentals} that when the statistical knowledges of received signals under both two hypotheses are fully known, the Neyman-Pearson-based likelihood ratio test (LRT) achieves the optimal detection performance while the energy detector represents a suboptimal choice compared to the LRT. However, for a system where such statistical knowledges are not available (e.g., the mmWave system concerned in this paper), LRT can hardly achieve the optimal performance while the energy detector can still provide an efficient detection performance.
\end{remark}

As to the feature weight settings, we first develop theoretical models of authentication performance metrics and then reveal the relationship between weights and beamforming gains of the beam signals. Based on these theoretical models, an optimization problem is formulated to determine the optimal weights. The detailed process is presented as follows.

\textbf{Theoretical Modeling of Authentication Performance Metrics.} Theoretical analysis of the proposed authentication scheme can facilitate the prediction and optimization of the security performance. Accordingly, we derive the analytical expressions of authentication metrics in terms of false alarm (denoted by $P_f$) and detection probabilities (denoted by $P_d$). The false alarm event refers to the receiver incorrectly identifying a legitimate device as an unauthorized one. Thus, $P_f$ is mathematically expressed as $P_f\triangleq\mathbf{Pr}(T>\tau|\mathcal{H}_0)$. On the other hand, a detection event is triggered when the receiver successfully identifies ongoing unauthorized attempts, that is  $P_d\triangleq\mathbf{Pr}(T>\tau|\mathcal{H}_1)$. To derive the closed-form expressions of $P_f$ and $P_d$, for a given $\tau$, it is necessary to investigate the right tail probabilities of test statistic $T$ under two hypotheses. 
However, since $T$ is a weighted sum of $L$ noncentral Chi-square variables, the PDF of $T$ cannot be directly acquired. Thus, we provide the following Lemma \ref{lemma_pdf} to approximate the right tail probability of $T$ under a threshold $\tau$.
\begin{lemma}\label{lemma_pdf}
    Let $Y_l$ denote the $l$-th noncentral chi-squared random variable with noncentrality parameter $\lambda_l$ and DoF 2. Consider a weighted sum of $L$ noncentral chi-squared random variables $T=\sum^{L}_{l=1}\omega_lY_l$, then the right-tail probability of $T$ with a threshold $\tau$ is approximated by a right-tail probability for a noncentral  chi-squared random variable $K\thicksim\chi^2_{\upsilon_K}(\lambda_K)$ under a new threshold $\tau_K$:
     \begin{align}
        \mathbf{Pr}\left(T>\tau\right)\approx \mathbf{Pr}\left(K>\tau_K\right),
    \end{align}
     where $\tau_K=[(\tau-\mu_T)/\sigma_T]\sqrt{2\upsilon_K+4\lambda_K}+\upsilon_K+\lambda_K$. $\mu_T=\gamma_1$ and $\sigma_T=\sqrt{2\gamma_2}$ denote the mean and standard deviation of $T$ with $\gamma_k=2\sum_{l=1}^L(\omega_l)^k+k\sum_{l=1}^{L}(\omega_l)^k\lambda_{l}$. Moreover, $\lambda_K$ and $\upsilon_K$ are respectively calculated by
     \begin{align}\label{lemma_lambda_nu}
        \lambda_K=s_1a^3-a^2,\ \ \ \ \upsilon_K=a^2-2\lambda_K,
    \end{align}
    where $a=1/(s_1-\sqrt{(s_1)^2-s_2})$ with $s_1=\gamma_3/(\gamma_{2})^{3/2}$ and $s_2=\gamma_4/(\gamma_2)^2$. 
\end{lemma}
\begin{proof}
    Please refer to \cite{liu2009new} for more details.
\end{proof}

Based on the above Lemma \ref{lemma_pdf}, the analytical expressions of $P_f$ and $P_d$ are summarized in Theorem \ref{theorem_2}.
\begin{theorem}\label{theorem_2}
    Given a fixed threshold $\tau$, $P_f$ and $P_d$ of the authentication scheme are given by, respectively
    \begin{align}
     P_f&\approx\int_{\tau_{A}}^{\infty}\frac{1}{2}\left(\frac{x}{\lambda_A}\right)^{\frac{(\upsilon_A-2)}{4}}\exp\left(-\frac{1}{2}\left(x+\lambda_A\right)\right)\nonumber\\
     &\ \ \ \times I_{\frac{\upsilon_A}{2}-1}\left(\sqrt{x\lambda_A}\right){\rm{d}}x,\label{pf_theo}\\
     P_d&\approx\int_{\tau_{E}}^{\infty}\frac{1}{2}\left(\frac{x}{\lambda_E}\right)^{\frac{(\upsilon_E-2)}{4}}\exp\left(-\frac{1}{2}\left(x+\lambda_E\right)\right)\nonumber\\
     &\ \ \ \times I_{\frac{\upsilon_E}{2}-1}\left(\sqrt{x\lambda_E}\right){\rm{d}}x,\label{pd_theo}
    \end{align}
    where $\tau_{K}=[(\tau-\mu_{T,K})/\sigma_{T,K}]\sqrt{2\upsilon_{K}+4\lambda_{K}}+\upsilon_{K}+\lambda_{K},K=\{A,E\}$. To avoid the repetition, the detailed calculation of parameters $(\mu_{T,K},\sigma_{T,K},\upsilon_{K},\lambda_{K})$ are explained in the following proof. $I_{\upsilon}(\cdot)$ is the modified Bessel function of the first kind with order $\upsilon=\frac{\upsilon_K}{2}-1$.
\end{theorem}
\begin{proof}
     See Appendix \ref{appendix_B} for the detailed derivation process.
\end{proof}
\vspace{-1.2em}
\textbf{Optimization of Weight Values.} The strategic allocation of optimal weights $\bm{\omega}=[\omega_1,\omega_2,\cdots,\omega_L]^T$ to $L$ beam pair signals $y_l,l\in[1,L]$ is crucial for significantly enhancing the authentication performance of CovertAuth. However, how to determine the optimal allocation strategy of $\bm{\omega}$ poses a challenge that remains unresolved. This is primarily because the interplay between weight values and authentication performance metrics has not been clearly revealed.
As a result, we first establish a functional relationship between $\bm{\omega}$ and $P_d$. Then, based on such a relationship, an optimization problem is formulated to acquire the optimal $\bm{\omega}$ that maximizes $P_d$ under a specified $P_f$ constraint.  The detailed optimization process is illustrated as follows. 

According to the Neyman-Pearson theorem, we generally evaluate the performance of $P_d$ under a maximum allowable false alarm probability $P_f=P_f^{\dagger}$. Let $Q_{\chi^2_{\upsilon}(\lambda)}(\tau)$ denote the right tail probability of a noncentral Chi-square variable with a threshold $\tau$ \cite{kay1993fundamentals}.  Then, the corresponding $\tau_{A}^{\dagger}$ and $\tau^{\dagger}$ are given by, respectively
\begin{align}
    \tau^{\dagger}_{A}&=Q^{-1}_{\chi^2_{\upsilon_A}(\lambda_A)}\left(P_f^{\dagger}\right),\label{cal_tauA}\\
    \tau^{\dagger}&=\frac{\tau_A-(\upsilon_A+\lambda_A)}{\sqrt{2\upsilon_A+4\lambda_A}}\sigma_{T,A}+\mu_{T,A}.\label{cal_taudagger}
\end{align}
For the given threshold $\tau^{\dagger}$, $P_d^{\dagger}$ under the constraint $P_f^{\dagger}$ is expressed by    
\begin{align}\label{link_pd_omega}
    P_d^{\dagger}=Q_{\chi^2_{\upsilon_E}(\lambda_E)}\left(\tau^{\dagger}_E\right),
\end{align}
where $\tau^{\dagger}_E=[(\tau^{\dagger}-\mu_{T,E})/\sigma_{T,E}]\sqrt{2\upsilon_E+4\lambda_E}+\upsilon_E+\lambda_E$. Note that the parameters of ($\upsilon_{K},\lambda_{K},\mu_{T,K}$ and $\sigma_{T,K}, K\in\{A,E\}$) are dependent on the $\bm{\omega}$. In other words, the equation \eqref{link_pd_omega} is a function of $\bm{\omega}$, implying that $\bm{\omega}$ directly influences $P_d^\dagger$. By exploiting this functional relationship, an optimization problem with the aim of finding the optimal $\bm{\omega}$ is formulated:
\vspace{-1.2em}
\begin{subequations}\label{opt}
  \begin{align}\label{solver_omega}
    \max _{\bm{\omega}}&\ \ P_d^{\dagger},\\
    \text{s.t.}&\  P_f=P_f^{\dagger},\\
    &\ \sum^{L}_{l=1}\omega_l=1,\\
    &\ 0<\omega_l<1.
\end{align}
\end{subequations}
It is observed that the optimization problem in \eqref{solver_omega} is non-convex due to the product of the exponential function and polynomial in the objective function. Thus, we employ the sequential quadratic programming (SQP) algorithm to solve it with the initialization $\bm{\omega}=\frac{1}{L}\mathbf{I}$ and the whole process is summarized in the Algorithm \ref{algorithm_2}.
\begin{algorithm}
	\caption{Feature Weight Allocation Algorithm.}
        \label{algorithm_2}
        \renewcommand{\algorithmicrequire}{ \textbf{Input:}}     
        \renewcommand{\algorithmicensure}{ \textbf{Output:}}  
	\begin{algorithmic}[1]  
		\REQUIRE $P_f^{\dagger}$.   
        \STATE \textbf{Initialization:}~ $\bm{\omega}=\frac{1}{L}\mathbf{I}$.
        \STATE Calculate $\lambda_A$ and $\upsilon_A$ using equation  \eqref{lemma_lambda_nu}.\label{step_2}
        \STATE Acquire $\tau_A^{\dagger}$ via equation \eqref{cal_tauA}.
        \STATE Compute $\tau^{\dagger}$ using equation \eqref{cal_taudagger}.
        \STATE Calculate $\lambda_E$ and $\upsilon_E$ following a similar way in \eqref{lemma_lambda_nu}.\label{step_6}
        \STATE Obtain $\tau_E^{\dagger}=[(\tau^{\dagger}-\mu_{T,E})/\sigma_{T,B}]\sqrt{2\upsilon_E+4\lambda_E}+\upsilon_E+\lambda_E$.
        \STATE Solve \eqref{opt} with SQP method to acquire $\bm{\omega}^{\star}$.
        \ENSURE  $\bm{\omega}^{\star}$.
	\end{algorithmic}
\end{algorithm}
\vspace{-2em}

\subsection{Convergence and Complexity Analysis}
\subsubsection{Convergence Analysis}Here, we provide an analysis on the convergence property of Algorithm \ref{Algothrim}.  First, the feasible sequences $\{P^t,t=1,…\}$ and $\{ N^t,t=1,…\}$ generated by iterations in Algorithm \ref{Algothrim} are both bounded and compacted, so a limiting point $\bm{\varsigma}^\star=\{P^\star,N^\star\}$ exists based on Algorithm \ref{Algothrim}. Second, the surrogate functions used in Algorithm \ref{Algothrim} satisfy the conditions of function value consistency, gradient consistency and lower bound, so we know from \cite{razaviyayn2014successive} that the objective function of Algorithm \ref{Algothrim} is non-decreasing and the limiting point $\bm{\varsigma}^\star$ satisfy the following Karush-Kuhn-Tucker (KKT) conditions:
\begin{align}
	\nabla\bm{f}(\bm{\varsigma}^{\star})-\nu^{\star}(N^{\star}P^{\star}\sum_{l=1}^{L}\left||\hat{\mathbf{h}}_e\mathbf{w}_l|+h_e|\mathbf{w}_l|\right|^2-2\sigma^2_e\epsilon^2)&=0,\nonumber\\
\nu^\star(N^{\star}P^{\star}\sum_{l=1}^{L}\left||\hat{\mathbf{h}}_e\mathbf{w}_l|+h_e|\mathbf{w}_l|\right|^2-2\sigma^2_e\epsilon^2)&=0,\nonumber\\
N^{\star}P^{\star}\sum_{l=1}^{L}\left||\hat{\mathbf{h}}_e\mathbf{w}_l|+h_e|\mathbf{w}_l|\right|^2-2\sigma^2_e\epsilon^2&\leq0,\nonumber\\
\nu^\star&\geq0,
\end{align}
where $\bm{f}(\bm{\varsigma}^{\star})=\log P_{\rm{a}}^{\star}+\log\left(\log\left(1+\frac{N^{\star}P^{\star}|\alpha|^2F_RW_T}{\sigma^2_n}\right)\right)+\log(1-\frac{N^{\star}}{N_{\rm{total}}})$ and $\nu^\star$ denotes the  dual optimal point. Thus, the limiting point $\bm{\varsigma}^\star$ is a KKT point of problem \eqref{problem_opt}. Based on the observations that the objective function of Algorithm \ref{Algothrim} is non-decreasing and the limiting point $\bm{\varsigma}^\star$ is a KKT point of problem \eqref{problem_opt}, we conclude that Algorithm \ref{Algothrim} converges.

\subsubsection{Complexity Analysis} Notice that the CovertAuth scheme is composed of two parts: the covert transmission phase design and the authentication decision. For the covert transmission phase design in CovertAuth, we observe that its computational complexity is primarily determined by the number of iterations $t_c$ and the computational complexity within each iteration. In particular, for each iteration, the complexity is dominated by updating $P^{t+1}$ in \eqref{solver_P} and $N^{t+1}$ in \eqref{solver_N}. Note that problem \eqref{solver_P} is a convex problem with a single variable and two affine constraints, while the popular interior point method is used in CVX to solve this problem, so the corresponding computational complexity is of order $\mathcal{O}{\{\ln(\frac{1}{\rho_p})\}}$, where $\rho_p$ denotes the accuracy of the tolerable duality gap to guarantee the optimality. Similarly, the problem \eqref{solver_N} has the same order of $\mathcal{O}{\{\ln(\frac{1}{\rho_p})\}}$. Thus, the computational complexity of the covert transmission phase is $\mathcal{O}{\{\ln(\frac{1}{\rho_p})t_c\}}$. For the authentication decision in \eqref{decision}, the computational complexity is of order $\mathcal{O}{\{L\}}$. Therefore, the total complexity of CovertAuth is $\mathcal{O}{\{\ln(\frac{1}{\rho_p})t_c+L\}}$.
\vspace{-0em}
\section{Numerical Results}\label{num results}
\subsection{Parameter Settings}
For the concerned mmWave BA stage, we consider that the transmitter sides of Alice and Eve are both equipped with an imperfect ULA array with $N_t=32$ elements. Moreover, the size of the transmitting beam codebook is set as $L_T=16$. 
As to the receiver Bob, the number of antenna elements is $N_r=16$ with a receiving beam codebook of size $L_R=8$. Consequently, the total number of candidate beam pairs is $L=128$. The total number of symbols within a time-slot is set as $N_{\rm{total}}=5210$.
In regard to the location of these communication entities, we assume that Alice and Eve are located around the receiver Bob. 
In specific, the AoA-AoD pair of Alice-Bob is set to $(\theta_0,\phi_0)=(30^\circ,60^\circ)$, while the AoA-AoD pair related to Eve-Bob is set as $(\theta_1,\phi_1)=(15^\circ,18^\circ)$. 
The corresponding channel gains associated with Alice and Eve are set as $\alpha_0\thicksim\mathcal{CN}(0,1)$ and $\alpha_1\thicksim\mathcal{CN}(0,2)$, respectively. Let $\kappa_n=|\alpha|^2/\sigma^2_n$ and $\kappa_e=|\alpha|^2_e/\sigma^2_e$ denote the pre-beamforming SNR at the Bob and Eve side, respectively. 
The MC effects of Alice and Eve are evaluated by $\sigma^2_{c,0}=0.1$ and $\sigma^2_{c,1}=0.4$, respectively.
The above parameters serve as the default values if not explicitly stated otherwise. We conduct 3000 Monte-Carlo trials on a PC with a 5.4 GHz Intel 14-Core i7 CPU and 32GB RAM to obtain average results.

\vspace{-0.8em}
\subsection{Covert Transmission Performance}\label{covert_performance}
\begin{figure*}[!t]
\vspace{-2.5em}
\subfigure[]
{\includegraphics[width=0.34\linewidth]{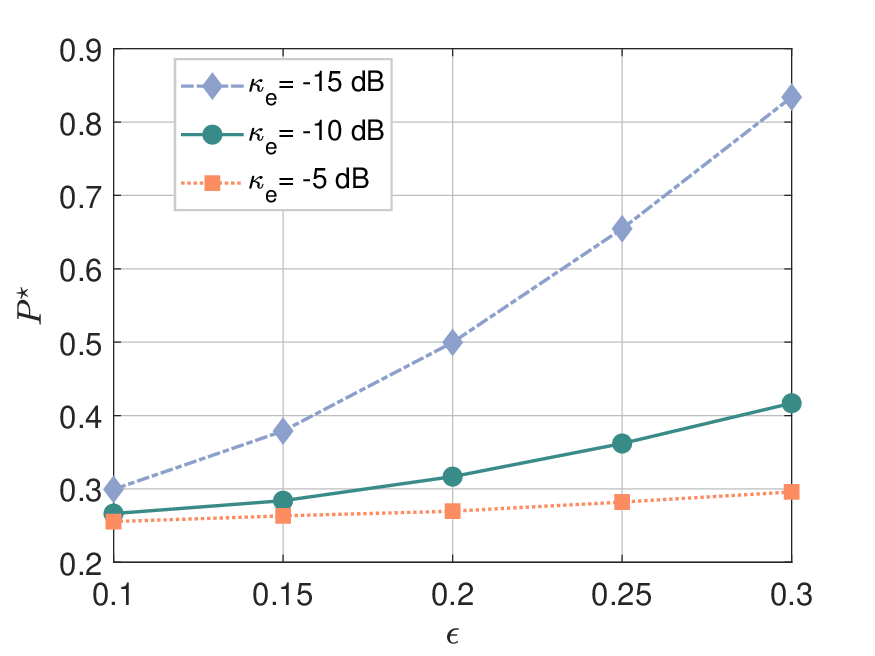}
\label{fig_epsilon_P}}
\hspace{-1.5em}
\subfigure[]
{\includegraphics[width=0.34\linewidth]{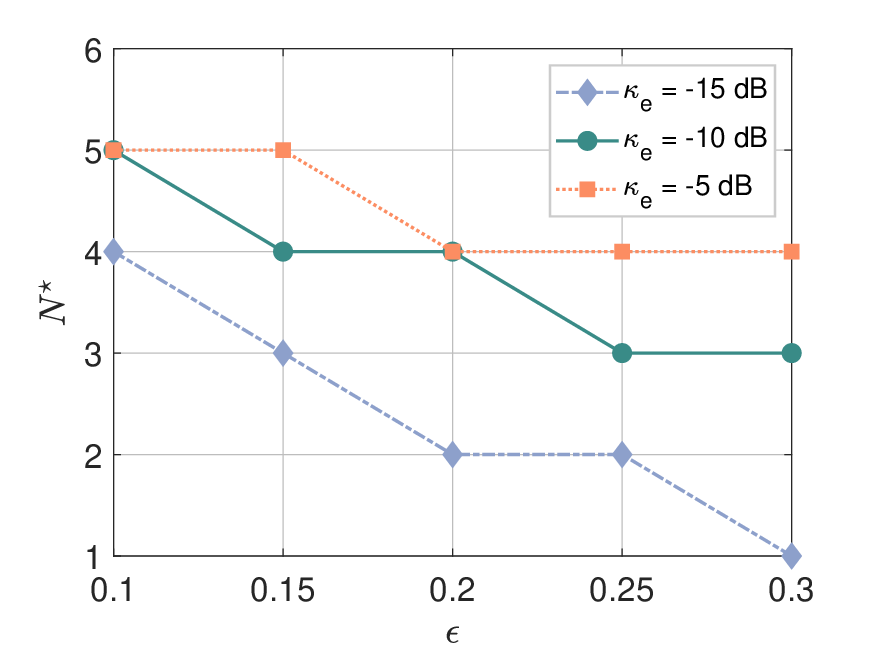}
\label{fig_epsilon_N}}
\hspace{-1.5em}
\subfigure[]
{\includegraphics[width=0.34\linewidth]{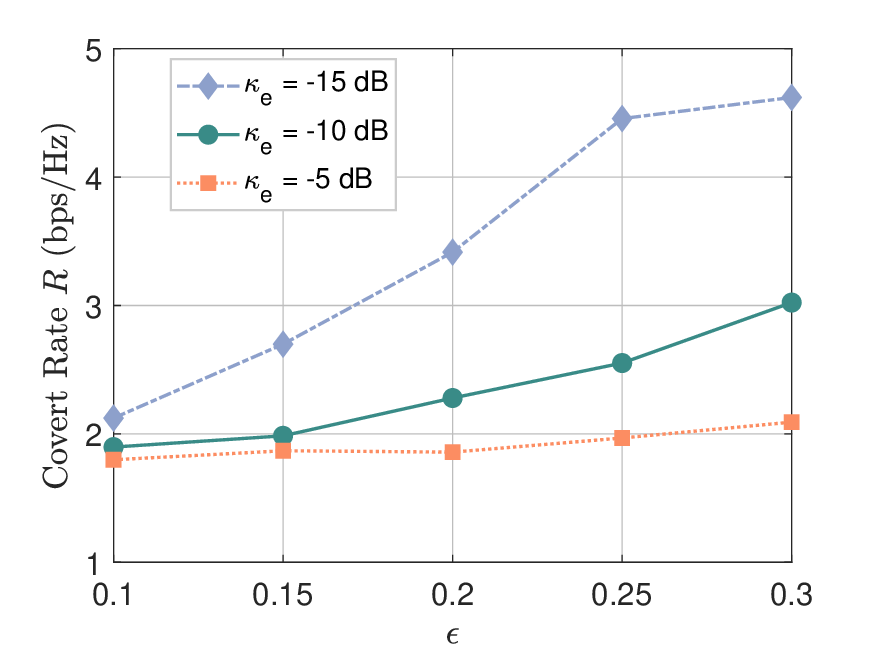}
\label{fig_epsilon_R}}
\setlength{\abovecaptionskip}{-0.5em}
\caption{The impact of covertness level $\epsilon$ on: (a) optimal transmitting power $P^{\star}$; (b) optimal beam training budget $N^{\star}$; (c) optimal covert rate $R$.}
\label{fig_epsilon_parameters}
\end{figure*}

\begin{figure*}[!t]
\vspace{-1.2em}
        \centering
        \hspace{-1em}
        \begin{minipage}{0.34\linewidth}
		\centering
        {\includegraphics[width=1\textwidth]{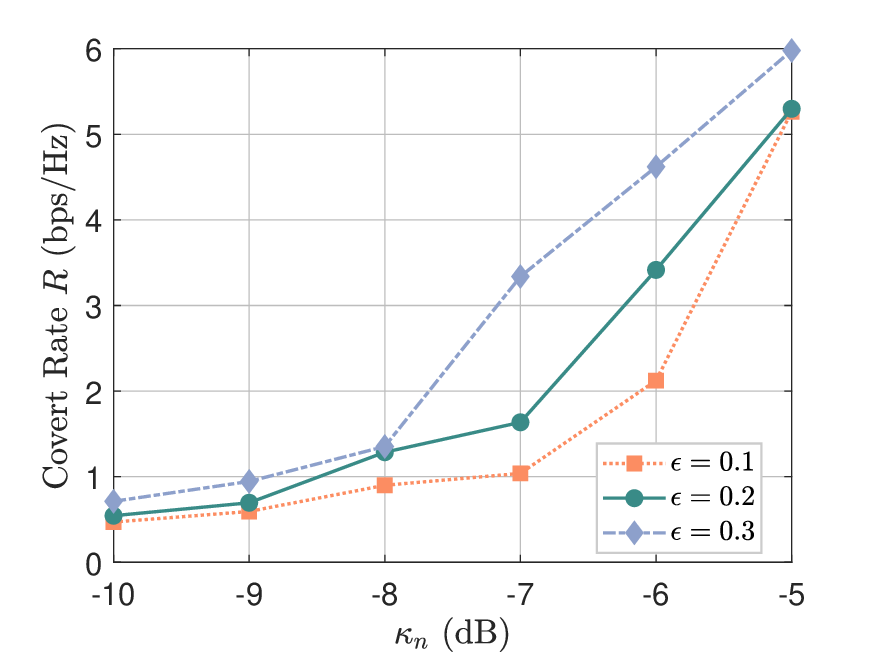}
        \setlength{\abovecaptionskip}{-0.5em}
        \caption{$R$ vs. $\kappa_n$ with different $\epsilon$.}
        \label{fig_SNR_R}}
	\end{minipage}
        \hspace{-0.7em}
        \begin{minipage}{0.34\linewidth}
		\centering
		\includegraphics[width=1\textwidth]{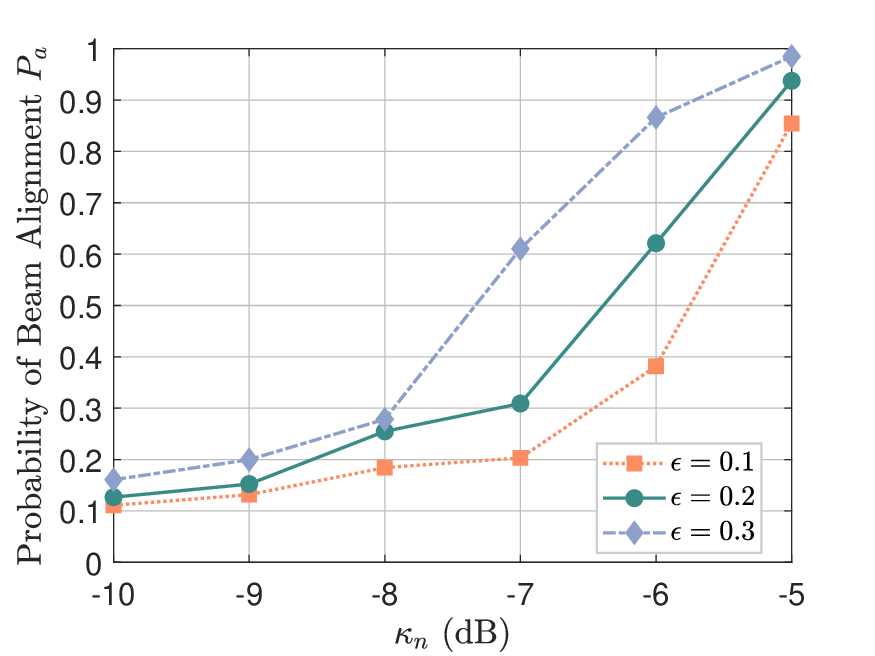}
  \setlength{\abovecaptionskip}{-0.5em}
		\caption{$P_{\rm{a}}$ vs. $\kappa_n$ with different $\epsilon$.}
		\label{fig_SNR_Pa}
	\end{minipage}
        \hspace{-1.4em}
        \begin{minipage}{0.34\linewidth}
		\centering
		\includegraphics[width=1\textwidth]{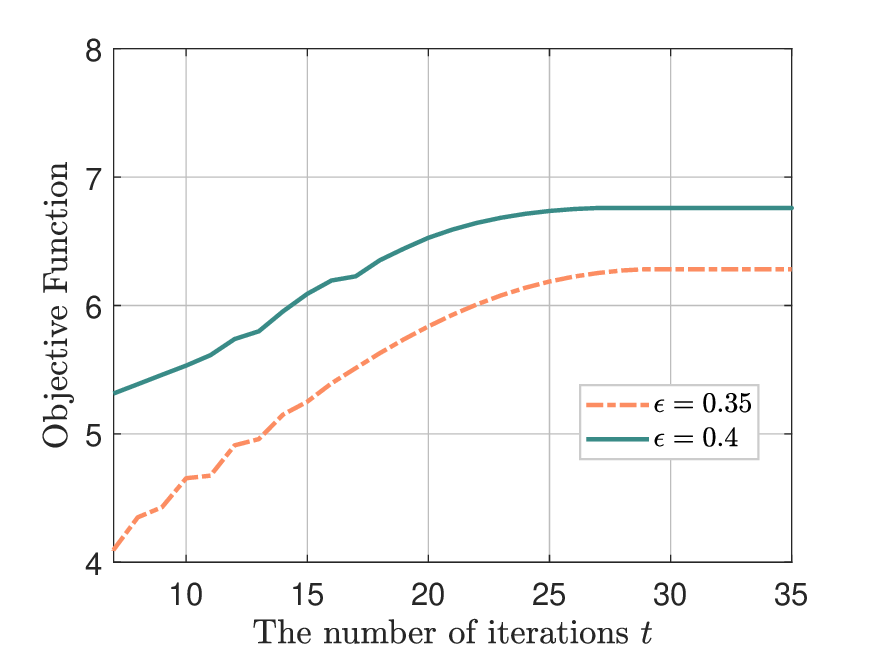}
  \setlength{\abovecaptionskip}{-0.5em}
		\caption{Illustration of the convergence of CovertAuth.}
		\label{fig_convergence}
	\end{minipage}
    \vspace{-1.2em}
\end{figure*}

In this part, we focus on evaluating the secure communication performance of the proposed CovertAuth scheme. In particular, we utilize the covertness level $\epsilon$ to characterize the secure communication demand. The secure communication performance of CovertAuth is evaluated by considering different impact parameters.

\textbf{Impact of $\epsilon$ on Covert Transmission Performance.} In Fig.~\ref{fig_epsilon_parameters}, we illustrate the impact of $\epsilon$ on optimized variables in terms of $P^{\star}$ and $N^{\star}$ as well as the resulting covert rate $R$ with the setting $\kappa_n=-6$ dB. As observed from Fig.~\ref{fig_epsilon_P}, a larger $\epsilon$ leads to an increased $P^{\star}$, implying that more transmitting power can be used in the BA stage as the covertness requirement is relaxed. Moreover, given a constant $\epsilon=0.2$, the optimized $P^{\star}$ under $\kappa_e=\{-5,-10,-15\}$ dB are 0.269, 0.316 and 0.499, respectively. The reason is the increased noise uncertainty level at Eve makes it difficult to successfully detect the covert transmission and thus Alice can transmit signals with a larger power $P^{\star}$. Another interesting observation in Fig.~\ref{fig_epsilon_N} is that a larger $P^\star$ allows for a reduction in beam training budget $N$. In concert, for $\kappa_e=-15$ dB, as $P^\star$ increases from 0.299 to 0.834, the optimization variable  $N^\star$ correspondingly decreases from 4 to 1. This indicates the trade-off between transmitting power and beam training budget needed to achieve the desired covertness constraint.
In Fig.~\ref{fig_epsilon_R}, we can see that the relaxed covertness requirement obviously results in a larger communication rate between Alice and Bob. For example, with $\kappa_e=-15$ dB, $R$ presents an increased tendency from 2.122 bps/Hz to 4.621 bps/Hz with $\epsilon$ in the range of $[0.1,0.3]$.

\textbf{Impact of $\kappa_n$ on Covert Transmission Performance.} 
Given the fixed $\kappa_e=-15$ dB, Fig.~\ref{fig_SNR_R} and Fig.~\ref{fig_SNR_Pa} present the effects of $\kappa_n$ on covert rate $R$ and probability of successful beam alignment $P_{\rm{a}}$ under different covertness constraints $\epsilon\in\{0.1,0.2,0.3\}$. In specific, one can note from Fig.~\ref{fig_SNR_R} that $R$ of all three curves exhibit an increasing tendency as $\kappa_n$ improves. This is because a lower noise interference in the Alice-Bob link enhances the channel transmission quality, thereby increasing the number of effective information bits that can be transmitted under the same covertness constraint. Also, with a constant $\kappa_n=-6$ dB, $R$ under three covertness constraints are $2.122$ bps/Hz, $3.415$ bps/Hz, and $4.621$ bps/Hz, respectively. We can conclude that the increased  $\epsilon$ (i.e., lower covertness requirement) can lead to higher $R$. As the demand for covertness diminishes, CovertAuth can afford to be more aggressive in data transmission, thus increasing the overall data rate. From Fig.~\ref{fig_SNR_Pa}, it is demonstrated that $P_{\rm{a}}$ can achieve over 0.85 when $\kappa_n$ is larger than $-5$ dB, ensuring the high transmission reliability of CovertAuth in low SNR conditions.
\vspace{-1.3em}

\textbf{Convergence Behavior of the CovertAuth Scheme.} To demonstrate the convergence of Algorithm \ref{Algothrim}, we plot in Fig.~\ref{fig_convergence} how the objective function in Algorithm \ref{Algothrim} varies with the number of iterations $t$ under the settings of $\kappa_n=-6$ dB, $\kappa_e=-15$ dB, $\epsilon=\{0.35, 0.4\}$. It can be seen that for a given setting of  $\epsilon$, the objective function first increases monotonically (e.g., within the first 20 iterations) and then converges to constants after a certain number of iterations (e.g., $t=25$ here).

\vspace{-1.2em}
\subsection{Identity Authentication Performance}
In this part, we evaluate the authentication performance of CovertAuth combating identity-based impersonation attacks.

\textbf{Theoretical Model Validation.} To validate the correctness of analytical expressions regarding $P_f$ in \eqref{pf_theo} and $P_d$ in \eqref{pd_theo}, we plot in Fig.~\ref{fig_auth_validate_PDF} - Fig.~\ref{fig_auth_validate_PD} the theoretical and simulated results about PDFs of $Y_l|\mathcal{H}_1$, $P_d$ and $P_f$, respectively. The relevant system parameters are set as ($\kappa_e=-15$ dB, $\epsilon=0.2$). Fig.~\ref{fig_auth_validate_PDF} presents that the theoretical PDF curves of $Y_l|\mathcal{H}_1$ are fitted well with the simulated ones under randomly selected beam signals $\{Y_{34}, Y_{35}, Y_{36}\}$. This implies the correct statistical analysis of $Y_l|\mathcal{H}_1$ and provides support for the subsequent derivation of $P_d$. It is observed from Fig.~\ref{fig_auth_validate_PF} and Fig.~\ref{fig_auth_validate_PD} that although the proposed theoretical models can efficiently characterize the authentication performance of the proposed scheme in terms of $P_f$ and $P_d$, a slight mismatch still exists between the theoretical results and corresponding simulated ones. Such a mismatch is mainly due to the approximation adopted in the theoretical modeling. Notice that to derive the closed-form expressions for the detection and false alarm probabilities, we formulate the test statistic $T$ as a weighted sum of noncentral Chi-square variables. However, it is difficult (if not impossible) to obtain the precise statistical information of $T$, so we apply Lemma \ref{lemma_pdf} to approximate $T$ as a noncentral Chi-square variable. While this approximation leads to a deviation between the theoretical and simulated results, we can observe that as SNR increases from $-6$ dB to 2 dB, the gaps between the theoretical and simulation results of $P_f$ and $P_d$ tend to decrease. Thus, the theoretical models in \eqref{pf_theo} and \eqref{pd_theo} can provide an efficient approximation to the simulated results, especially in scenarios with low noise interference.

\textbf{Performance Comparison with Different Weight Allocation Schemes.}
To exhibit the performance superiority of CovertAuth using the optimal weight allocation strategy obtained in problem \eqref{opt}, we compare it with the one exploiting the averaged weight design scheme and summarize the results in Fig.~\ref{fig_auth_weight} with the receiver operating characteristic (ROC) curve. Specifically, we fix the $P_f$ between 0.05 and 0.5 with a step size of 0.05 and calculate the corresponding $P_d$ under the settings of ($\kappa_e=-15$ dB, $\epsilon=0.2$, $\kappa_n=\{-6,-4\}$ dB). It is noted that for a fixed $P_f=0.1$, the scheme with optimal weight value can achieve 24.7\% detection performance improvement over that with average weight value under $\kappa_n=-6$ dB. This indicates that the solved optimal weight enables CovertAuth to obtain a larger detection accuracy gain and to possess more reliability in countering impersonation attacks.

\textbf{Impact of $\epsilon$ on Authentication Performance.} To investigate the trade-off performance between covert transmission and identity authentication, we examine the impact of $\epsilon=\{0.1,0.2,0.3\}$ on $P_d$ and $P_f$ in Fig.~\ref{fig_auth_epsilon} with the settings of ($\kappa_n=-6$ dB, $\kappa_e=-15$ dB). One can see that a higher $\epsilon$ results in a lower miss detection probability $P_m=1-P_d$ and a lower $P_f$. However, the increase in authentication performance often comes with the sacrifice of performance in combating eavesdropping attacks, thus a balance between covert transmission and identity authentication should be achieved based on the current security requirements of the mmWave system. Based on the observed results, we can note that CovertAuth enables effective defense against identity impersonation attacks while preserving covert communication. According to different security level requirements, CovertAuth can dynamically adjust $\epsilon$ to balance the trade-off between covertness and authentication performance. This dynamic flexibility makes CovertAuth well-suited for privacy-sensitive environments, military communications, and secure IoT networks, where both robust identity verification and undetectable transmission are essential.

\textbf{Impact of $\kappa_e$ on Authentication Performance.} To explore the performance of CovertAuth under different attack levels, Fig.~\ref{fig_auth_kappa_w} illustrates the effect of $\kappa_e$ on $P_m$ and $P_f$ with the settings of ($\kappa_n=-6$ dB, $\epsilon=0.2$). One can note that the metric of $P_m$ and $P_f$ both present an increasing tendency as $\kappa_e$ improves. The reason is that as the noise level at Eve decreases, Eve gains an improved rate to accurately analyze and interpret the transmitted signals, making it easier to mimic legitimate users during the authentication phase. This results in a degraded authentication performance. Consequently, if the priority of identity authentication is high, it may be necessary to appropriately compromise the secrecy of information transmission to improve authentication accuracy.

\textbf{Robustness Illustration of CovertAuth.} To further examine the reliability of CovertAuth, we consider a worst-case scenario, where the attacker Eve is located along the transmission path between Alice and Bob. Under the considered scenario, Eve possesses the same AoA-AoD pair as that in the Alice-Bob link. In addition, we further assume that the channel gains $\alpha_0$ and $\alpha_1$ maintain the same value during the BA stage. Then, we investigate the performance of the proposed scheme under different $\epsilon$ with ($\kappa_n =-$6 dB and $\kappa_e=-$15 dB). The covert transmission performance is summarized in Table \ref{covert_table} and the corresponding authentication performance is plotted in Fig.~\ref{fig_auth_worst}. Although a low $\epsilon$ may constrain the authentication performance, CovertAuth can demonstrate a remarkable enhancement in authentication accuracy when the covertness constraints are moderately relaxed. In specific, under a given $P_f=0.1\%$, $P_m$ under $\epsilon=\{0.1,0.2,0.3\}$ are 58.3\%, 7.2\%, and 0.03\%, respectively. This suggests that CovertAuth can still achieve a relatively satisfactory performance even under the worst-case scenario.
\begin{figure*}[!t]
\vspace{-2em}
    \subfigure[]
    {\includegraphics[width=0.34\linewidth]{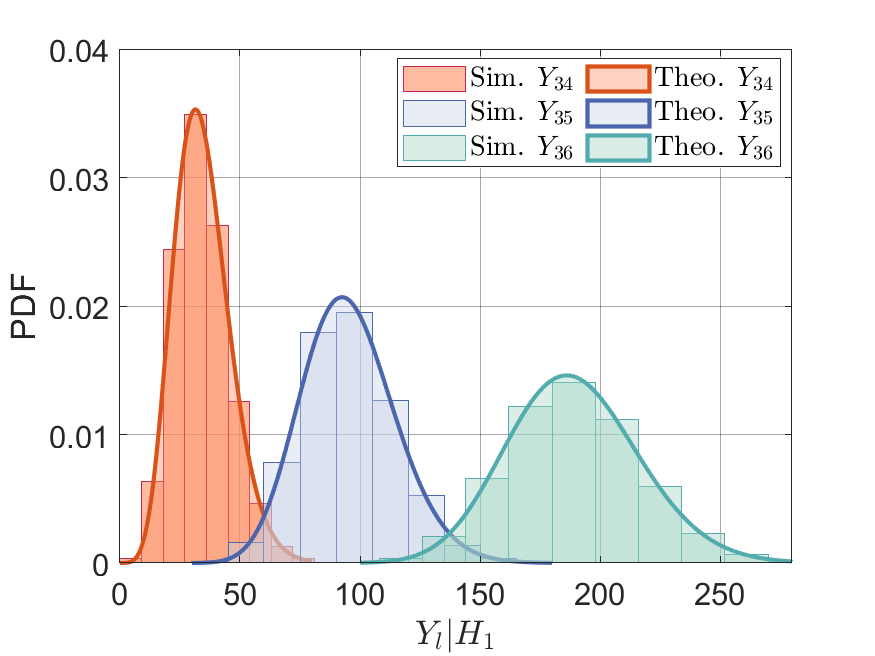}
    \label{fig_auth_validate_PDF}}
    \hspace{-1.5em}
    \subfigure[]
    {\includegraphics[width=0.34\linewidth]{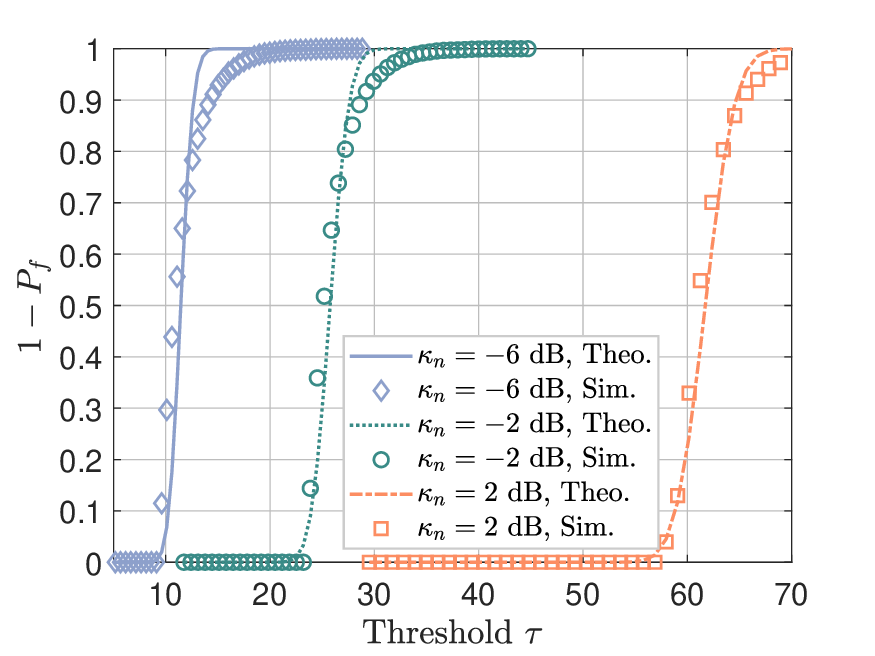}
    \label{fig_auth_validate_PF}}
    \hspace{-1.5em}
    \subfigure[]
    {\includegraphics[width=0.34\linewidth]{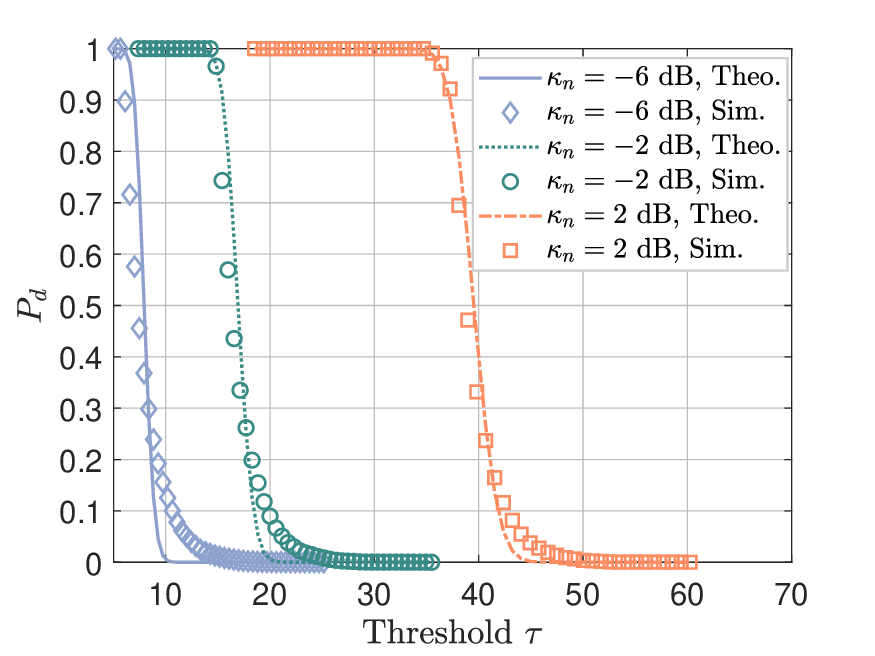}
    \label{fig_auth_validate_PD}}
    \setlength{\abovecaptionskip}{0em}
    \caption{Comparison between theoretical and simulation results of (a) PDF of $Y_l|\mathcal{H}_1$; (b) $P_f$ versus $\tau$; (c) $P_d$ versus $\tau$.}
    \label{fig_auth_validate}
    \vspace{-1.5em}
\end{figure*}

\begin{figure*}[!t]
\hspace{-3mm}
	\begin{minipage}{0.242\linewidth}
 \vspace{0em}
		\centering
        \includegraphics[width=1.12\textwidth]{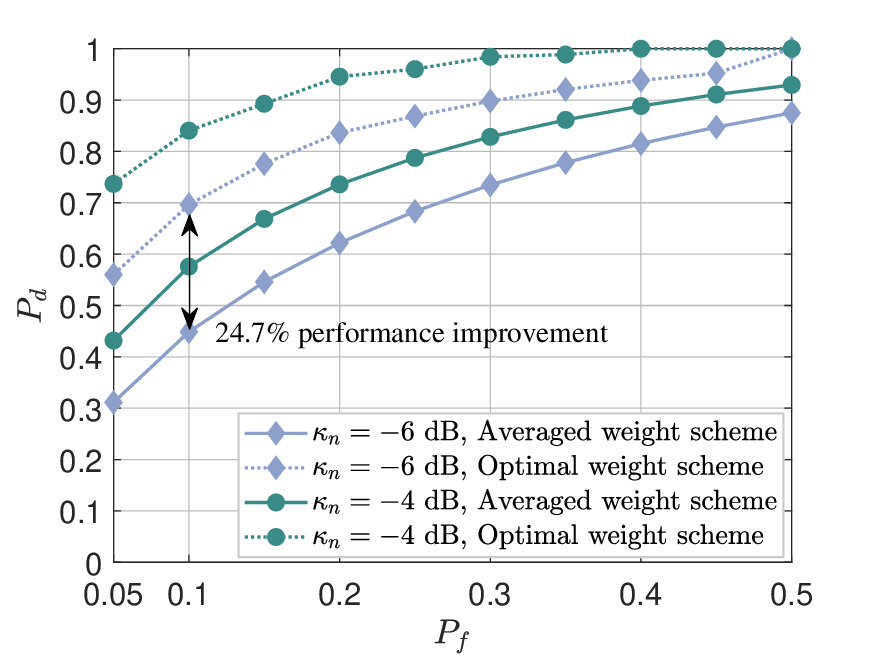}
        \setlength{\abovecaptionskip}{-1.5em}
        \caption{ROC curves under different weight design schemes.}
        \label{fig_auth_weight}
	\end{minipage}
 \hspace{-0.06em}
        \centering
        \begin{minipage}{0.242\linewidth}
		\centering
        \setlength{\abovecaptionskip}{-0.5em}
		\includegraphics[width=1.14\linewidth]{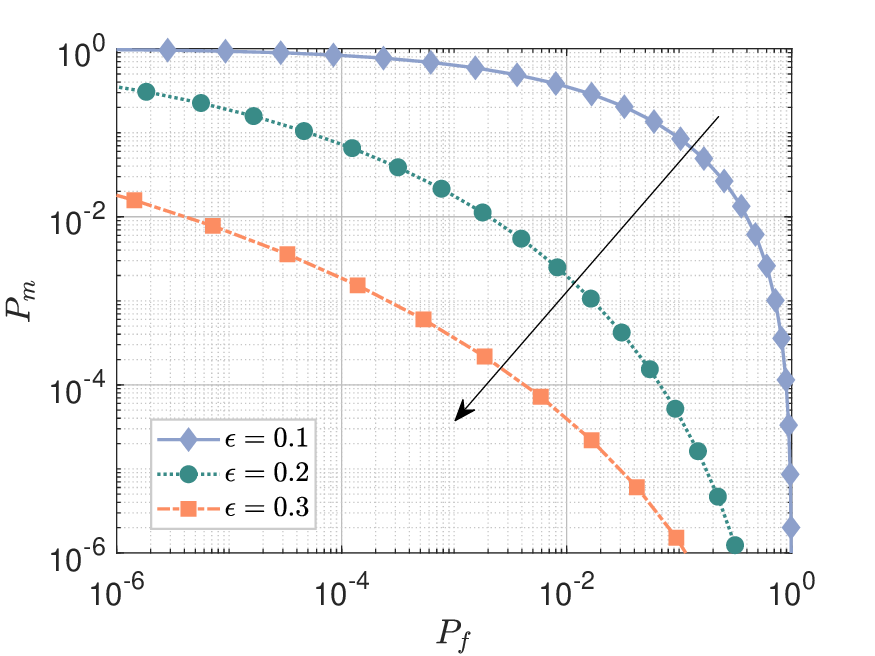}
		\caption{Impact of $\epsilon$ on authentication performance.}
		\label{fig_auth_epsilon}
	\end{minipage}
  \hspace{0.18em}
	\begin{minipage}{0.242\linewidth}
 \vspace{0.5em}
		\centering
          \setlength{\abovecaptionskip}{-0.6em}
		\includegraphics[width=1.115\linewidth]{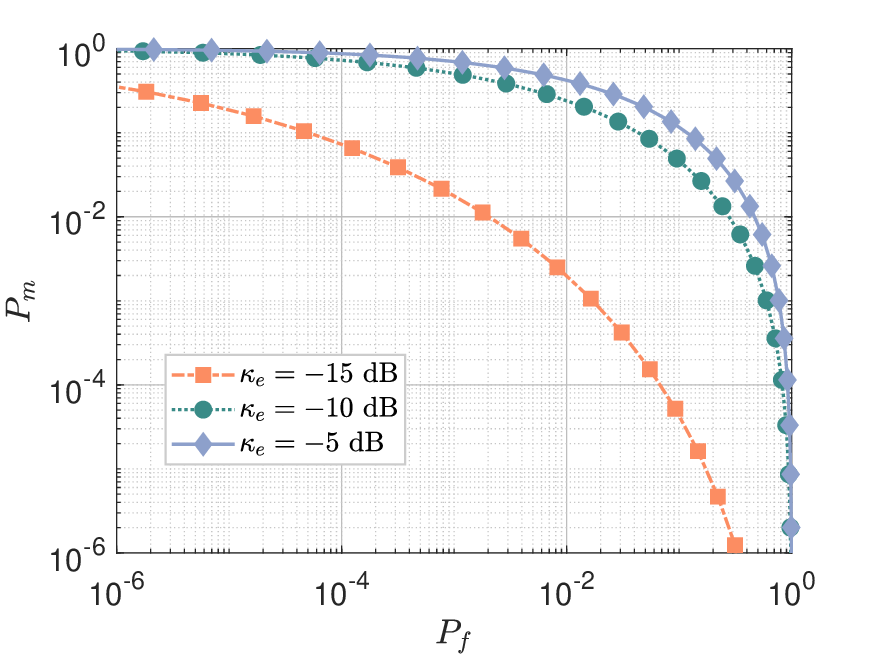}
		\caption{Impact of $\kappa_e$ on authentication performance.}
		\label{fig_auth_kappa_w}
	\end{minipage}	
  \hspace{-0.1em}
         \centering
        \begin{minipage}{0.243\linewidth}
		\centering
          \setlength{\abovecaptionskip}{-0.6em}
		\includegraphics[width=1.12\linewidth]{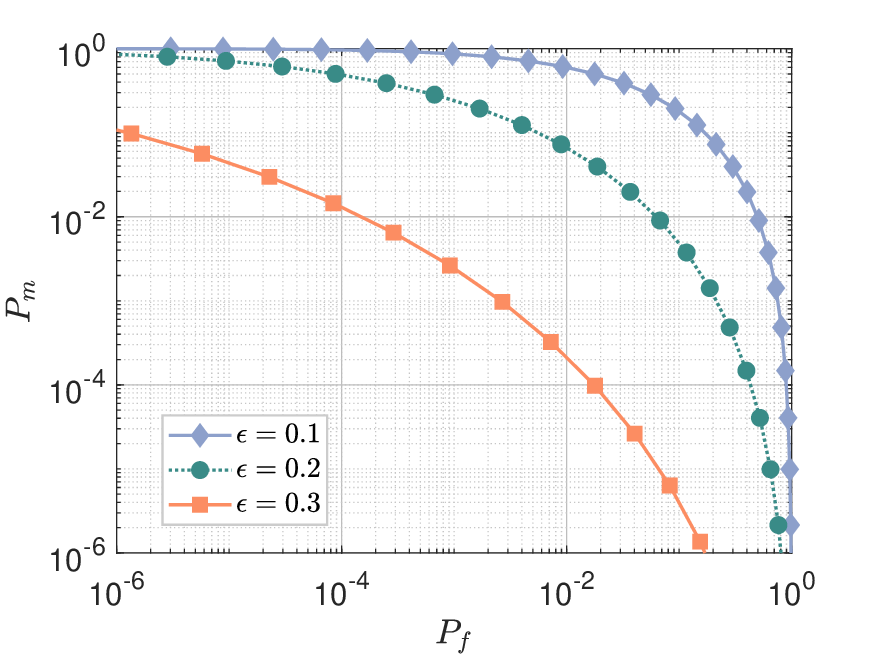}
		\caption{Authentication performance under worst-case scenario.}
		\label{fig_auth_worst}
	\end{minipage}
 \vspace{-1.5em}
\end{figure*}
\begin{figure}
    \subfigure[]
    {\includegraphics[width=0.92\linewidth]{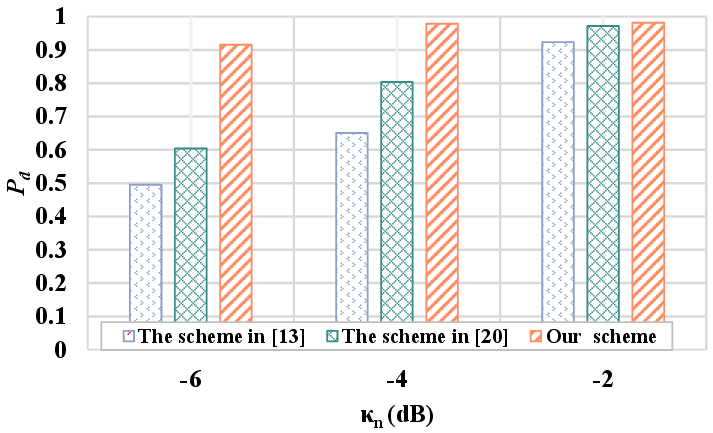}
    \label{fig_auth_comparison_kappa}}
    \vspace{-1em}
    \\
    \subfigure[]
    {\includegraphics[width=0.92\linewidth]{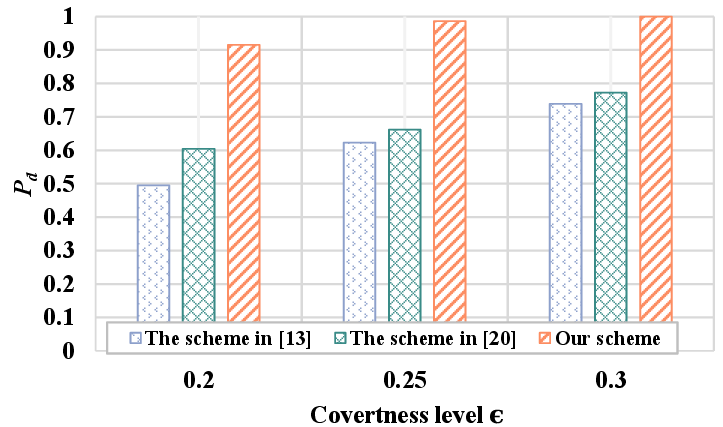}
    \label{fig_auth_comparison_epsilon}}
    \setlength{\abovecaptionskip}{-0em}
    \caption{Authentication performance comparison with existing works under the impacts of (a) $\kappa_n$; (b) Covertness level $\epsilon$.}
    \vspace{-1.5em}
\end{figure}
\begin{table}[htbp]
\vspace{-1.2em}
\renewcommand{\arraystretch}{1.3} 
\centering
\caption{Covert Communication Performance under the Worst-case Scenario}
\begin{tabular}{cccc}
\hline Covertness level $\epsilon$ & 0.1 & 0.2 & 0.3 \\
\hline $P^\star$ &0.2758 & 0.3659   &  0.5288 \\
$N^\star$ &5  & 3   &  2 \\
$R$ (bps/Hz)& 1.9125 & 2.3611 &3.6276 \\
\hline
\end{tabular}
\label{covert_table}
\end{table}

\textbf{Performance Comparison with Existing Works.} In Fig.~\ref{fig_auth_comparison_kappa} and Fig.~\ref{fig_auth_comparison_epsilon}, we present a comparative analysis with existing works of channel phase response-based scheme in \cite{DBLP:journals/tifs/LuLSL23} and multiple channel responses-based scheme in \cite{DBLP:journals/ton/LiZCCXL24} under the impacts of $\kappa_n$ and $\epsilon$, respectively. From Fig.~\ref{fig_auth_comparison_kappa}, one can note that all three methods present an increased $P_d$ due to the decreased noise interference at Bob. 
Although the proposed CovertAuth is comparable to the method in \cite{DBLP:journals/ton/LiZCCXL24} under high SNR conditions (i.e., $\kappa_n=-2$ dB), our scheme exhibits the best detection performance in low SNR regime (e.g., $\kappa_n=\{-6,-4\}$ dB). This indicates that the proposed scheme not only delivers superior detection capabilities but also exhibits robustness against noise interference. Specifically, given a fixed $\kappa_n=-4$ dB, CovertAuth can provide 32\% and 17\% performance improvement compared to the methods in \cite{DBLP:journals/tifs/LuLSL23} and \cite{DBLP:journals/ton/LiZCCXL24}, respectively. 

Fig.~\ref{fig_auth_comparison_epsilon} illustrates the authentication performance considering the covertness constraint $\epsilon=\{0.2,0.25,0.3\}$. We can intuitively see that the performance of the three schemes gradually improves with the relaxation of covertness constraints, and CovertAuth exhibits the best performance under all conditions. This is because under a low $\epsilon$, the allocated transmission power $P^{\star}$ is limited, leading to poor feature extraction precision and degraded authentication performance of existing works. In contrast, CovertAuth exploits the fine-grained beam pattern feature caused by both the antenna hardware fingerprint MC effect and the spatial information of the transmitter to devise an energy detector-based authentication scheme without the need for feature extraction. Such a scheme effectively meets the security requirements in the mmWave BA stage. 

\section{Discussions}\label{Discussions}
Notice that the above analysis is based on the ideal beam pattern without any side-lobe leakage. In this section, we further investigate the performance of the proposed CovertAuth scheme under a more practical beam pattern model with side-lobe leakage and various number of antennas at transmitting/receiving nodes.
\vspace{-1.4em}
\subsection{New Beam Pattern Model with Side-lobe Leakage}
We consider a new beam pattern model where the training beams exhibit uniform gain within the main-lobe and a constant small leakage in the side-lobe \cite{DBLP:journals/tifs/ZhangLZJX22}, \cite{DBLP:journals/twc/LiLHCW19}. Then, with this assumption, the beamforming gains of the main-lobe and side-lobe at the transceiver sides can be quantified as
\begin{align}
W_l(\phi) & = \begin{cases}W_T , & \text { if } \phi \in \Phi_{\mathbf{w}_l}, \\
w_T, & \text { otherwise},\end{cases} \\
F_l(\theta) & = \begin{cases}F_R, & \text { if } \theta \in \Theta_{\mathbf{f}_l}, \\
f_R, & \text { otherwise},\end{cases}
\end{align}
where $w_T$ and $f_R$ denote the beamforming gain when the beam misalignment occurs at the transmitter or receiver side, respectively. In this scenario, four possible beam alignment events may happen (i.e., perfect alignment at both the transmitter and receiver, misalignment at transmitter only, misalignment at receiver only, and misalignment at both ends), and the corresponding effective channel gain of the beam pattern $g_l$ will be  $|\alpha|^2F_RW_T,|\alpha|^2F_Rw_T,|\alpha|^2f_RW_T,$ and $ |\alpha|^2f_Rw_T$, respectively.
\vspace{-1.6em}
\subsection{Performance Analysis under New Beam Pattern Model}\label{covert_phase_practical}
\subsubsection{Covert Communication Performance} We first derive the new expression for the lower bound of successful beam alignment probability under the new beam pattern model. If we use
$P_{\rm{ma,1}}$, $P_{\rm{ma,2}}$, and $P_{\rm{ma,3}}$ to denote the probabilities of misalignment under the  three misalignment events, then the lower bound $P_{\rm{a,LB}}$ of $P_{\rm{a}}$ can be written as
\begin{align}\label{P_a_LB}
    P_{\rm{a,LB}}=1-P_{\rm{ma,1}}-P_{\rm{ma,2}}-P_{\rm{ma,3}},
\end{align}
where
\begin{align}
P_{\rm{ma},1} &= 1 - \int_0^{\infty} \left( F(y \mid 2, \lambda_2) \right)^{L_T - 1} f(y \mid 2, \lambda_1) \, {\rm{d}}y,  \\
P_{\rm{ma},2} &= 1 - \int_0^{\infty} \left( F(y \mid 2, \lambda_3) \right)^{L_R - 1} f(y \mid 2, \lambda_1) \, {\rm{d}}y, \\
P_{\rm{ma},3} &= 1 - \int_0^{\infty} \left( F(y \mid 2, \lambda_4) \right)^{(L_T - 1)(L_R - 1)} f(y \mid 2, \lambda_1) \, {\rm{d}}y.
\end{align}
$\lambda_2=\frac{2|\alpha|^2NPF_Rw_T}{\sigma^2_n}$, $\lambda_3=\frac{2|\alpha|^2NPf_RW_T}{\sigma^2_n}$, and $\lambda_4=\frac{2|\alpha|^2NPf_Rw_T}{\sigma^2_n}$ are the non-centrality parameters
related to the side-lobe effect. Moreover, $F(y \mid 2, \lambda)$ and $f(y \mid 2, \lambda)$ represent the cumulative distribution function (CDF) and probability density function (PDF) of a non-central Chi-square variable with non-centrality parameter $\lambda$ and degrees of freedom 2. Readers may refer to \cite{DBLP:journals/tifs/ZhangLYLCZW21}  for the detailed derivation of $P_{\rm{a,LB}}$, which is omitted here. Due to the fact that the beamforming gain of the main-lobe is generally much larger than that of the side-lobe, we can approximate the average effective covert rate $R$ as follows by ignoring the marginal contribution from misalignment:
\begin{align}\label{covert_rate_practical}
    R&=\left(1-\frac{N}{N_{\rm{total}}}\right)\mathbb{E}_{g_{\hat{l}}}\left\{\log\left(1+\frac{NPg_{\hat{l}}}{\sigma^2_n}\right)\right\}.\nonumber\\
    &\approx\left(1-\frac{N}{N_{\rm{total}}}\right)P_{\rm{a,LB}}\log\left(1+\frac{|\alpha|^2NPF_RW_T}{\sigma^2_n}\right).
\end{align}
Under the new beam pattern model, the related optimization problem for covert rate maximization can be formulated as
\begin{subequations}\label{problem_opt_practical}
\begin{align}
\max _{P, N}& \ \left(1-\frac{N}{N_{\rm{total}}}\right)P_{\rm{a,LB}}\log\left(1+\frac{|\alpha|^2NPF_RW_T}{\sigma^2_n}\right),  \\ 
\text { s.t. }  &0<P \leq P_{\text {max }}, \\
& 1 \leq N \leq  N_{\text {max }}, \\
& \eqref{covert_constraint}, \eqref{imperfect_CSI_constraint}.\nonumber
\end{align}
\end{subequations}
We can see that the key distinction between this optimization problem \eqref{problem_opt_practical} and problem in \eqref{problem_opt} is the incorporation of side-lobe beamforming gain in the successful beam alignment probability $P_{\rm{a,LB}}$. Consequently, we can still apply Algorithm \ref{Algothrim} to solve the problem \eqref{problem_opt_practical} by replacing $P_{\rm{a}}$ in \eqref{P_a_P} and \eqref{P_a_N} with $P_{\rm{a,LB}}$ in \eqref{P_a_LB}.
\subsubsection{Authentication Performance}
Note that the identity decision of the energy-based authentication mechanism in \eqref{decision} is actually determined by the received beam pattern energy rather than the quantified beamforming gain. In other words, the effects of side-lobe have already been incorporated into our design of the authentication mechanism. Therefore, under the new beam pattern model, the theoretical modeling and optimization framework of the authentication mechanism in Section \ref{authentication_phase} still work. However, with the new beam pattern model, the optimal parameter settings for covert communication (e.g., $N$ and $P$) will differ, leading to a different received signal energy and thus a different authentication performance. Such effects will be further illustrated in Section \ref{simulation_paractical}.

\subsection{Numerical Results under the New Beam Pattern Model}\label{simulation_paractical}
To demonstrate the effectiveness of CovertAuth under the new beam pattern model, we first quantify the beamforming gain of side-lobe leakage.  Specifically, we follow \cite{DBLP:journals/tifs/ZhangLYLCZW21} to calculate the beamforming gains of the side-lobe at the transceiver sides as
\begin{align}
    w_T &=\left(2 - ( 2/L_T \cdot W_T )\right)/(2 - 2/L_T), \\
f_R& = \left(2 - ( 2/L_R \cdot W_R )\right)/(2 - 2/L_R).
\end{align}

\textbf{Covert Transmission Performance.} 
To explore the effect of side-lobe on covert communication performance, we show in Fig.~\ref{fig_opt_R_practical} how $R$ varies with $\kappa_n$ under the ideal scenario without side-lobe leakage and practical scenario with side-lobe leakage, where $\kappa_e=-15$ dB and $\epsilon=\{0.1, 0.3\}$. It is observed that for a given $\kappa_n$ and covertness constraint $\epsilon$, $R$ under the ideal scenario is always larger than that under the practical scenario. The reason is that the side-lobe leakage degrades the precision of beam alignment and thus leads to a reduction in covert rate $R$. We can also see from Fig.~\ref{fig_opt_R_practical} that even with the consideration of side-lobe leakage, the covert rate can reach over 6 bps/Hz when SNR is larger than 0 dB, so our scheme can still satisfy both the communication quality and covertness requirements even in a practical scenario.

To further illustrate the covert communication performance gap between the ideal and practical scenarios, we show in Fig.~\ref{fig_opt_R_antenna} how covert rate varies with $\kappa_n$ under the settings of $\kappa_e=-15$ dB, $\epsilon=0.1$, ($N_t=16, N_r=8$), ($N_t=32, N_r=16$), and ($N_t=64, N_r=32$). As depicted in Fig.~\ref{fig_opt_R_antenna}, employing more antennas will yield a higher $R$. The reason is that employing a larger antenna array helps to enhance the beamforming capability and enables the construction of more precise and directive beams with better-controlled side-lobes, leading to an improved covert rate. Moreover, we notice that as more antennas are deployed, the performance gap between ideal and practical scenarios tend to decrease, indicating that the proposed scheme is more appealing for applications with a large-scale antenna array deployment. However, deploying more antennas will invariably lead to a higher hardware cost and a higher computational complexity for beam management, so a careful trade-off between the overhead and performance gain should be initialized for different applications.
\begin{figure*}[!t]
\vspace{-2em}
\hspace{-3mm}
	\begin{minipage}{0.242\linewidth}
 \vspace{0em}
		\centering
        \includegraphics[width=1.12\textwidth]{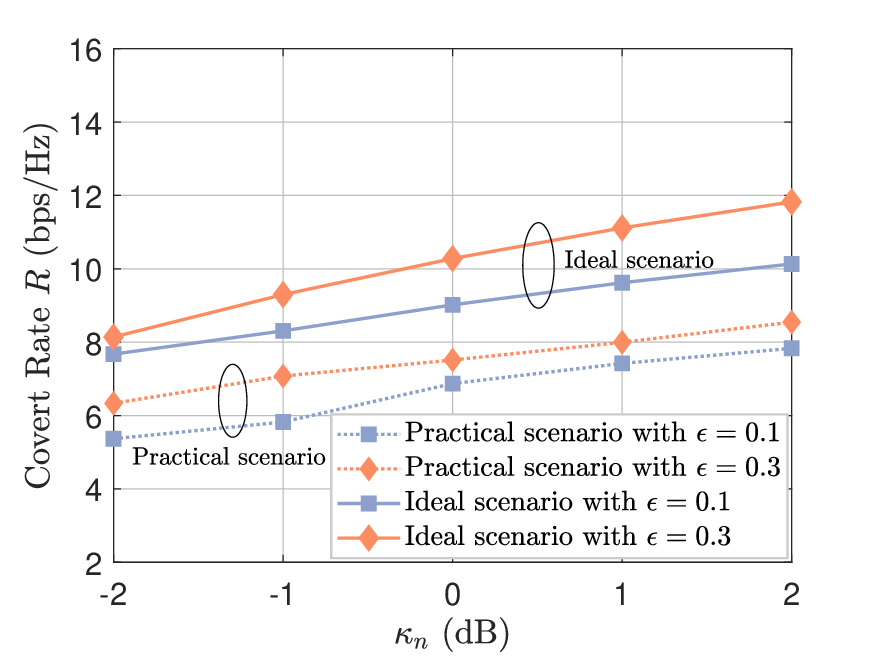}
        \setlength{\abovecaptionskip}{-1.5em}
        \caption{$R$ vs. $\kappa_n$ under different covertness constraints.}
        \label{fig_opt_R_practical}
	\end{minipage}
 \hspace{-0.06em}
        \centering
        \begin{minipage}{0.242\linewidth}
		\centering
        \setlength{\abovecaptionskip}{-0.5em}
		\includegraphics[width=1.14\linewidth]{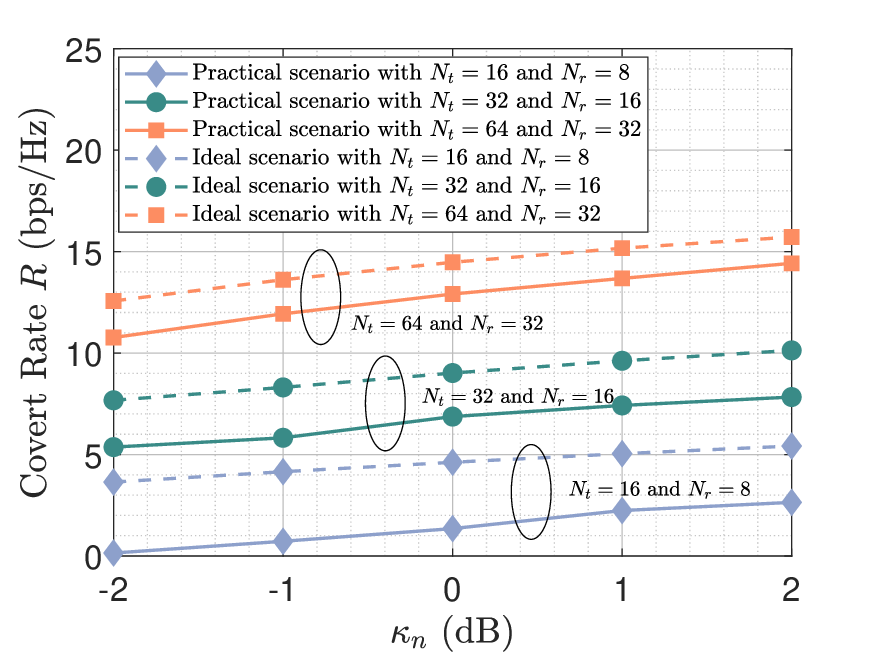}
		\caption{$R$ vs. $\kappa_n$ under different antenna array configurations.}
		\label{fig_opt_R_antenna}
	\end{minipage}
  \hspace{0.18em}
	\begin{minipage}{0.242\linewidth}
 \vspace{0.5em}
		\centering
          \setlength{\abovecaptionskip}{-0.6em}
		\includegraphics[width=1.12\linewidth]{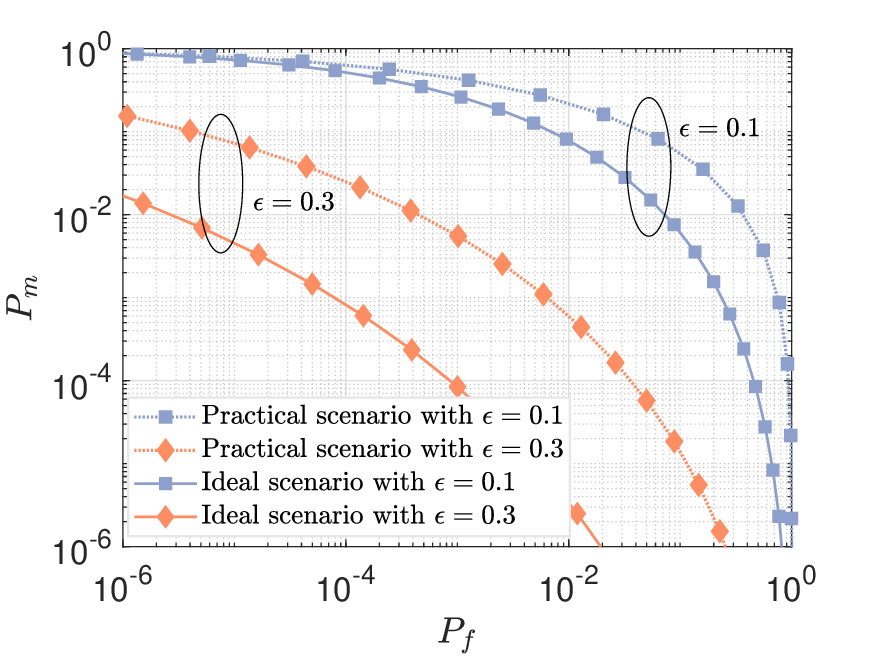}
		\caption{$P_m$ vs. $P_f$ under different covertness constraints.}
		\label{fig_auth_epsilon_practical}
	\end{minipage}	
  \hspace{-0.1em}
         \centering
        \begin{minipage}{0.243\linewidth}
		\centering
          \setlength{\abovecaptionskip}{-0.6em}
		\includegraphics[width=1.14\linewidth]{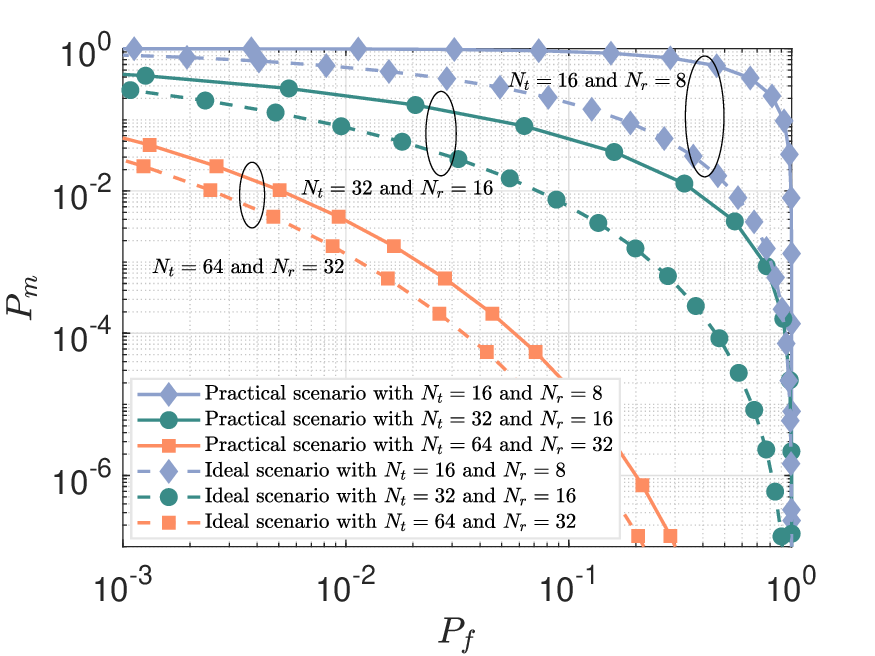}
		\caption{$P_m$ vs. $P_f$ under different antenna array configurations.}
		\label{fig_auth_antenna_practical}
	\end{minipage}
 \vspace{-1em}
\end{figure*}

\textbf{Identity Authentication Performance.}
To explore the effects of side-lobe on authentication performance, we illustrate in Fig.~\ref{fig_auth_epsilon_practical} $P_m$ vs. $P_f$ under the ideal scenario without side-lobe leakage and practical scenario with side-lobe leakage, where $\kappa_n=-2$ dB, $\kappa_e=-15$ dB, $\epsilon=\{0.1,0.3\}$. It can be seen from Fig.~\ref{fig_auth_epsilon_practical} that for a given $P_f$, the $P_m$ performance under the practical scenario is larger than that under the ideal scenario. This is because the unexpected side-lobe results in a decrease in signal transmission power and a distortion of the MC feature, resulting in a degraded authentication capability at the energy detector. The results in Fig.~\ref{fig_auth_epsilon_practical} also indicate that under a practical scenario, a slight compromise on the covertness requirement can lead to a significant improvement in authentication performance.

To further explore the authentication performance gap between the ideal and practical scenarios, we present in Fig.~\ref{fig_auth_antenna_practical} how the antenna array configuration influences $P_m$ and $P_f$ under the settings of $\kappa_n=-2$ dB, $\kappa_e=-15$ dB, $\epsilon=0.1$. As observed from Fig.~\ref{fig_auth_antenna_practical} that as the number of antennas increases, the authentication performance under the practical scenario approaches that under the ideal scenario. The reason behind such observation is similar to that of Fig.~\ref{fig_opt_R_antenna}, i.e., the enhanced directive beamforming enables an increase in the effective signal power towards the receiver, improving the authentication capability. We can also see from Fig.~\ref{fig_auth_antenna_practical} that under the settings of $(N_t=32,N_r=16)$, $P_m$ of the practical scenario achieves 0.054 when $P_f=0.1$, indicating that our scheme with such configuration is capable of meeting the authentication requirement of mmWave BA stage even under the effect of side-lobe.

\section{Related Work}\label{related_work}
Some preliminary efforts have been devoted to the secure mechanism design for the mmWave BA stage, which can be categorized into approaches addressing jamming attacks \cite{DBLP:journals/tcom/DarsenaV22}, \cite{DBLP:journals/tifs/DinhVanHCFMH23}, eavesdropping attacks \cite{DBLP:journals/tifs/ZhangLYLCZW21}, \cite{DBLP:journals/icl/XingQCWLL23}, \cite{DBLP:journals/tce/JuZBLPWO23}, and impersonation attacks \cite{DBLP:conf/mobicom/SteinmetzerAAH018}, \cite{10735492}.

\textbf{Countermeasures for Jamming Attacks.} To mitigate the BA misalignment from jamming attacks, a countermeasure exploiting randomized probing technique has been proposed in \cite{DBLP:journals/tcom/DarsenaV22}, where the base station corrupts the probing sequence randomly to enable jamming rejection at the user equipment using subspace-based techniques, such as orthogonal projection and jamming cancellation. Then, the authors in \cite{DBLP:journals/tifs/DinhVanHCFMH23} use an autoencoder-based approach for jamming detection and mitigation, enabling recovery of the corrupted received signal strength vector during the BA stage.

\textbf{Countermeasures for Eavesdropping Attacks.} As to the countermeasure for eavesdropping attacks, Zhang \textit{et. al} introduce a covert beam training strategy to minimize detection probability while enabling simultaneous beam training for a single user \cite{DBLP:journals/tifs/ZhangLYLCZW21}.
Moreover, a two-stage power optimization scheme is developed in \cite{DBLP:journals/icl/XingQCWLL23} to maximize the average covert rate for the BA and data transmission phase while achieving the covertness requirement to avoid eavesdropping threat. With the random beam switching aid,  Ju \textit{et. al} exploit the virtual angles of effective and activated beams as the random source to generate the secret key for information security between the transceiver pair \cite{DBLP:journals/tce/JuZBLPWO23}. 

This paper proposes a CovertAuth scheme to integrate a signal optimization-based covert communication mechanism with an adaptive weight-based authentication mechanism to simultaneously combat both eavesdropping and impersonation attacks. While the covert communication mechanism in this work is highly related to that of \cite{DBLP:journals/tifs/ZhangLYLCZW21}, some primary distinctions exist. First, notice that the hardware imperfections (e.g., mutual coupling effects) in antenna arrays can lead to beam pattern distortion and beam squint, we incorporate such mutual coupling effects into the beamforming design to achieve a more efficient covert transmission. Second, due to channel estimation errors and the non-cooperative relationship between Alice and Eve, this work considers a more practical imperfect CSI model in the overall covert signal optimization to achieve a more robust covert communication.

\textbf{Countermeasures for Impersonation Attacks.} To counter forged feedback from malicious attackers in the BA stage, the authors in \cite{DBLP:conf/mobicom/SteinmetzerAAH018} establish a session secret through asymmetric key exchange, and then appends a cryptographic nonce to the beam feedback, thus ensuring devices accept feedback only from authorized devices. Following this line, Li \textit{et. al} propose a novel secure beam sweeping protocol, named SecBeam, which leverages power/sector randomization techniques along with coarse angle-of-arrival information to effectively detect beam-stealing attacks during the BA phase \cite{10735492}. 

\section{Conclusion}\label{conclusion}
This paper proposed an innovative PLS framework named CovertAuth to effectively address eavesdropping and identity-based impersonation attacks for the BA stage in mmWave communication systems. CovertAuth developed a covert communication optimization framework to jointly design the beam training budget and transmission power for maximizing the covert communication rate while satisfying the covertness requirement during the BA phase. Moreover, it exploited the beam pattern feature impacted by the MC effects to achieve identity validation. Simulation results indicate: 1) incorporating the MC effect into the beam pattern enhances the authentication reliability of CovertAuth; 2) the derived theoretical models provide a valuable framework for authentication performance characterization and optimization; and 3) a trade-off between authentication and covert communication performance is observed under different security requirements. It is anticipated that CovertAuth can provide insightful guidelines for the design of a secure BA framework in mmWave communication systems.
\vspace{-1.2em}
\appendices
\section{Proof of Theorem \ref{theorem_2}}\label{appendix_B}
Under $\mathcal{H}_0$, let $Y_l=\frac{2|y_l|^2}{N^\star\sigma^2_n}$ and we can easily observe that $Y_l|\mathcal{H}_0$ is distributed to a noncentral Chi-square distribution with DoF 2 and noncentral parameter $\lambda_{l,0}=\frac{2N^\star P^\star|\mathbf{f}^{H}_l\mathbf{H}_0\mathbf{w}_l|^2}{\sigma^2_n}$, where $\mathbf{H}_0=\alpha_0\Tilde{\mathbf{a}}_{r,0}(\theta_0)\Tilde{\mathbf{a}}_{t,0}^{H}(\phi_0)$. Thus, the test statistic $T$ is a weighted sum of $L$ noncentral Chi-squared variables $Y_l$ with $\lambda_{l,0}$ and DoF 2. From Lemma \ref{lemma_pdf}, $P_f$ can be approximated by the right tail probability of a noncentral Chi-square variable $K_A$ with noncentral parameter $\lambda_{A}$ and DoF $\upsilon_A$:
     \begin{align}
         P_f&\triangleq\mathbf{Pr}(T>\tau|\mathcal{H}_0)\nonumber\\
            &\approx \mathbf{Pr}(K_A>\tau_A),
     \end{align}
     where $\tau_{A}=[(\tau-\mu_{T,A})/\sigma_{T,A}]\sqrt{2\upsilon_{A}+4\lambda_{A}}+\upsilon_{A}+\lambda_{A}$. The parameters $\mu_{T,A}$ and $\sigma_{T,A}$ are given by, respectively
     \begin{align}
        \mu_{T,A}&=\gamma_{A,1},\\
        \sigma_{T,A}&=\sqrt{2\gamma_{A,2}},
    \end{align}
     where the $k$-th cumulant of $T$ is calculated as $\gamma_{A,k}=2\sum_{l=1}^L(\omega_l)^k+k\sum_{l=1}^{L}(\omega_l)^k\lambda_{l,0}$. In addition, the new noncentral parameter $\lambda_{A}$ and DoF $\upsilon_{A}$ are calculated by
     \begin{align}\label{cal_A}
         \lambda_{A}=s_{A,1}a_{A}^3-a_{A}^2,\ \ \ \ \upsilon_A=a_{A}^2-2\lambda_{A},
     \end{align}
     where $a_A=1/(s_{A,1}-\sqrt{(s_{A,1})^2-s_{A,2}})$ with $s_{A,1}=\gamma_{A,3}/(\gamma_{A,2})^{3/2}$ and $s_{A,2}=\gamma_{A,4}/(\gamma_{A,2})^2$. Based on the PDF of non-central Chi-square distribution, we can achieve the final result in \eqref{pf_theo}.
     
     Under alternative hypothesis $\mathcal{H}_1$, $Y_l|\mathcal{H}_1\thicksim\chi^2_2(\lambda_{l,1})$ with $\lambda_{l,1}=\frac{2N^\star P^\star|\mathbf{f}^{H}_l\mathbf{H}_1\mathbf{w}_l|^2}{\sigma^2_n}$, where $\mathbf{H}_1=\alpha_1\Tilde{\mathbf{a}}_{r,1}(\theta_1)\Tilde{\mathbf{a}}_{t,1}^{H}(\phi_1)$. Similar to the derivation of $P_f$, $P_d$ can also be approximated by a right tail probability of noncentral Chi-square variable $K_E$ with parameters $\lambda_E$ and DoF $\upsilon_E$. Following a similar manner in \eqref{cal_A}, we can calculate $\lambda_E$ as well as $\upsilon_E$, and the theoretical expression of $P_d$ can be obtained in \eqref{pd_theo}.
     
\bibliographystyle{IEEEtran}
\bibliography{CovertAuth}

\begin{thebibliography}{10}
\providecommand{\url}[1]{#1}
\csname url@samestyle\endcsname
\providecommand{\newblock}{\relax}
\providecommand{\bibinfo}[2]{#2}
\providecommand{\BIBentrySTDinterwordspacing}{\spaceskip=0pt\relax}
\providecommand{\BIBentryALTinterwordstretchfactor}{4}
\providecommand{\BIBentryALTinterwordspacing}{\spaceskip=\fontdimen2\font plus
\BIBentryALTinterwordstretchfactor\fontdimen3\font minus \fontdimen4\font\relax}
\providecommand{\BIBforeignlanguage}[2]{{%
\expandafter\ifx\csname l@#1\endcsname\relax
\typeout{** WARNING: IEEEtran.bst: No hyphenation pattern has been}%
\typeout{** loaded for the language `#1'. Using the pattern for}%
\typeout{** the default language instead.}%
\else
\language=\csname l@#1\endcsname
\fi
#2}}
\providecommand{\BIBdecl}{\relax}
\BIBdecl

\bibitem{DBLP:journals/comsur/TanLGWPZZL24}
J.~Tan, T.~H. Luan, W.~Guan, Y.~Wang, H.~Peng, Y.~Zhang, D.~Zhao, and N.~Lu, ``Beam alignment in {mmWave} {V2X} communications: {A} survey,'' \emph{{IEEE} Commun. Surv. Tutorials}, vol.~26, no.~3, pp. 1676--1709, 2024.

\bibitem{DBLP:conf/wisec/SteinmetzerYH18}
D.~Steinmetzer, Y.~Yuan, and M.~Hollick, ``Beam-stealing: Intercepting the sector sweep to launch man-in-the-middle attacks on wireless {IEEE} 802.11ad networks,'' in \emph{Proc. 11th ACM Conf. Secur. Privacy Wirel. Mobile Netw.}, 2018, pp. 12--22.

\bibitem{DBLP:journals/twc/YangZTS24}
J.~Yang, W.~Zhu, M.~Tao, and S.~Sun, ``Hierarchical beam alignment for millimeter-wave communication systems: {A} deep learning approach,'' \emph{{IEEE} Trans. Wirel. Commun.}, vol.~23, no.~4, pp. 3541--3556, 2024.

\bibitem{DBLP:conf/mobicom/SteinmetzerAAH018}
D.~Steinmetzer, S.~Ahmad, N.~A. Anagnostopoulos, M.~Hollick, and S.~Katzenbeisser, ``Authenticating the sector sweep to protect against beam-stealing attacks in {IEEE} 802.11ad networks,'' in \emph{Proc. 24th Annu. Int. Conf. Mobile Comput. Netw. Workshops (MobiCom WKSHPS)}, 2018, pp. 3--8.

\bibitem{DBLP:journals/tifs/QiuCZ23}
B.~Qiu, W.~Cheng, and W.~Zhang, ``Robust multi-beam secure {mmWave} wireless communication for hybrid wiretapping systems,'' \emph{{IEEE} Trans. Inf. Forensics Secur.}, vol.~18, pp. 1393--1406, 2023.

\bibitem{DBLP:journals/tmc/XieZZTLN24}
N.~Xie, J.~Zhang, Q.~Zhang, H.~Tan, A.~X. Liu, and D.~Niyato, ``Hybrid physical-layer authentication,'' \emph{{IEEE} Trans. Mob. Comput.}, vol.~23, no.~2, pp. 1295--1311, 2024.

\bibitem{DBLP:journals/ton/YuYL23}
K.~Yu, J.~Yu, and C.~Luo, ``The impact of mobility on physical layer security of {5G} {IoT} networks,'' \emph{{IEEE/ACM} Trans. Netw.}, vol.~31, no.~3, pp. 1042--1055, 2023.

\bibitem{DBLP:journals/ton/XieTHL21}
N.~Xie, H.~J. Tan, L.~Huang, and A.~X. Liu, ``Physical-layer authentication in wirelessly powered communication networks,'' \emph{{IEEE/ACM} Trans. Netw.}, vol.~29, no.~4, pp. 1827--1840, 2021.

\bibitem{DBLP:journals/ton/LiSKWLLZ23}
J.~Li, G.~Sun, H.~Kang, A.~Wang, S.~Liang, Y.~Liu, and Y.~Zhang, ``Multi-objective optimization approaches for physical layer secure communications based on collaborative beamforming in {UAV} networks,'' \emph{{IEEE/ACM} Trans. Netw.}, vol.~31, no.~4, pp. 1902--1917, 2023.

\bibitem{DBLP:journals/pieee/JiangWCS24}
Y.~Jiang, L.~Wang, H.~Chen, and X.~Shen, ``Physical layer covert communication in {B5G} wireless networks - its research, applications, and challenges,'' \emph{Proc. {IEEE}}, vol. 112, no.~1, pp. 47--82, 2024.

\bibitem{DBLP:conf/infocom/WangJWLZ20}
N.~Wang, L.~Jiao, P.~Wang, W.~Li, and K.~Zeng, ``Machine learning-based spoofing attack detection in {mmWave} {60GHz} {IEEE} 802.11ad networks,'' in \emph{Proc. IEEE Conf. Comput. Commun. (INFOCOM)}, 2020, pp. 2579--2588.

\bibitem{DBLP:journals/tifs/NosouhiSGD22}
M.~R. Nosouhi, K.~Sood, M.~Grobler, and R.~Doss, ``Towards spoofing resistant next generation iot networks,'' \emph{{IEEE} Trans. Inf. Forensics Secur.}, vol.~17, pp. 1669--1683, 2022.

\bibitem{DBLP:journals/tifs/LuLSL23}
X.~Lu, J.~Lei, Y.~Shi, and W.~Li, ``Physical-layer authentication based on channel phase responses for multi-carriers transmission,'' \emph{{IEEE} Trans. Inf. Forensics Secur.}, vol.~18, pp. 1734--1748, 2023.

\bibitem{DBLP:journals/tifs/TengZCJX24}
Y.~Teng, P.~Zhang, X.~Chen, X.~Jiang, and F.~Xiao, ``{PHY}-layer authentication exploiting channel sparsity in {mmWave} {MIMO} {UAV-Ground} systems,'' \emph{{IEEE} Trans. Inf. Forensics Secur.}, vol.~19, pp. 4642--4657, 2024.

\bibitem{DBLP:journals/tifs/ZhangLZJX22}
J.~Zhang, M.~Li, M.~Zhao, X.~Ji, and W.~Xu, ``Multi-user beam training and transmission design for covert millimeter-wave communication,'' \emph{{IEEE} Trans. Inf. Forensics Secur.}, vol.~17, pp. 1528--1543, 2022.

\bibitem{DBLP:journals/twc/JamaliM22}
M.~V. Jamali and H.~Mahdavifar, ``Covert millimeter-wave communication: Design strategies and performance analysis,'' \emph{{IEEE} Trans. Wirel. Commun.}, vol.~21, no.~6, pp. 3691--3704, 2022.

\bibitem{DBLP:journals/twc/WangLN22}
C.~Wang, Z.~Li, and D.~W.~K. Ng, ``Covert rate optimization of millimeter wave full-duplex communications,'' \emph{{IEEE} Trans. Wirel. Commun.}, vol.~21, no.~5, pp. 2844--2861, 2022.

\bibitem{DBLP:journals/twc/XiaoHLWSWY24}
H.~Xiao, X.~Hu, A.~Li, W.~Wang, Z.~Su, K.~Wong, and K.~Yang, ``{STAR-RIS} enhanced joint physical layer security and covert communications for multi-antenna {mmWave} systems,'' \emph{{IEEE} Trans. Wirel. Commun.}, vol.~23, no.~8, pp. 8805--8819, 2024.

\bibitem{DBLP:journals/jsac/XieZCT22}
N.~Xie, Q.~Zhang, J.~Chen, and H.~Tan, ``Privacy-preserving physical-layer authentication for non-orthogonal multiple access systems,'' \emph{{IEEE} J. Sel. Areas Commun.}, vol.~40, no.~4, pp. 1371--1385, 2022.

\bibitem{DBLP:journals/ton/LiZCCXL24}
Y.~Li, J.~Zhang, J.~Chen, Y.~Chen, N.~Xie, and H.~Li, ``Privacy-preserving physical-layer authentication under cooperative attacks,'' \emph{{IEEE/ACM} Trans. Netw.}, vol.~32, no.~2, pp. 1171--1186, 2024.

\bibitem{DBLP:journals/tap/Schmid13}
C.~M. Schmid, S.~Schuster, R.~Feger, and A.~Stelzer, ``On the effects of calibration errors and mutual coupling on the beam pattern of an antenna array,'' \emph{{IEEE} Trans. Antennas Propag.}, vol.~61, no.~8, pp. 4063--4072, 2013.

\bibitem{DBLP:journals/twc/ZhengWLJWWY22}
T.~Zheng, Y.~Wen, H.~Liu, Y.~Ju, H.~Wang, K.~Wong, and J.~Yuan, ``Physical-layer security of uplink mmwave transmissions in cellular {V2X} networks,'' \emph{{IEEE} Trans. Wirel. Commun.}, vol.~21, no.~11, pp. 9818--9833, 2022.

\bibitem{DBLP:journals/twc/HuangHYS25}
Y.~Huang, Y.~Hu, X.~Yuan, and A.~Schmeink, ``Analytical optimal joint resource allocation and continuous trajectory design for {UAV}-assisted covert communications,'' \emph{{IEEE} Trans. Wirel. Commun.}, vol.~24, no.~1, pp. 213--227, 2025.

\bibitem{DBLP:journals/tsp/AubryMLR23}
A.~Aubry, A.~D. Maio, L.~Lan, and M.~Rosamilia, ``Adaptive radar detection and bearing estimation in the presence of unknown mutual coupling,'' \emph{{IEEE} Trans. Signal Process.}, vol.~71, pp. 1248--1262, 2023.

\bibitem{DBLP:journals/tsp/ChenCCW19}
P.~Chen, Z.~Cao, Z.~Chen, and X.~Wang, ``Off-grid {DOA} estimation using sparse bayesian learning in {MIMO} radar with unknown mutual coupling,'' \emph{{IEEE} Trans. Signal Process.}, vol.~67, no.~1, pp. 208--220, 2019.

\bibitem{DBLP:journals/tcom/ZhangHZY20}
J.~Zhang, Y.~Huang, Y.~Zhou, and X.~You, ``Beam alignment and tracking for millimeter wave communications via bandit learning,'' \emph{{IEEE} Trans. Commun.}, vol.~68, no.~9, pp. 5519--5533, 2020.

\bibitem{DBLP:journals/tcom/ZhangHSWY17}
J.~Zhang, Y.~Huang, Q.~Shi, J.~Wang, and L.~Yang, ``Codebook design for beam alignment in millimeter wave communication systems,'' \emph{{IEEE} Trans. Commun.}, vol.~65, no.~11, pp. 4980--4995, 2017.

\bibitem{DBLP:journals/tifs/ZhangLYLCZW21}
J.~Zhang, M.~Li, S.~Yan, C.~Liu, X.~Chen, M.~Zhao, and P.~Whiting, ``Joint beam training and data transmission design for covert millimeter-wave communication,'' \emph{{IEEE} Trans. Inf. Forensics Secur.}, vol.~16, pp. 2232--2245, 2021.

\bibitem{DBLP:journals/tsp/ZhaoCSCZ16}
M.~Zhao, Y.~Cai, Q.~Shi, B.~Champagne, and M.~Zhao, ``Robust transceiver design for {MISO} interference channel with energy harvesting,'' \emph{{IEEE} Trans. Signal Process.}, vol.~64, no.~17, pp. 4618--4633, 2016.

\bibitem{boyd2004convex}
S.~Boyd and L.~Vandenberghe, \emph{{Convex Optimization}}.\hskip 1em plus 0.5em minus 0.4em\relax Cambridge university press, 2004.

\bibitem{DBLP:journals/tvt/HuangW18}
K.~Huang and H.~Wang, ``Combating the control signal spoofing attack in {UAV} systems,'' \emph{{IEEE} Trans. Veh. Technol.}, vol.~67, no.~8, pp. 7769--7773, 2018.

\bibitem{DBLP:journals/iotj/HeNZQ24}
J.~He, M.~Niu, P.~Zhang, and C.~Qin, ``Enhancing {PHY}-layer authentication in {RIS}-assisted {IoT} systems with cascaded channel features,'' \emph{{IEEE} Internet Things J.}, vol.~11, no.~14, pp. 24\,984--24\,997, 2024.

\bibitem{DBLP:journals/tdsc/ZhangTNJF24}
P.~Zhang, Y.~Teng, M.~Niu, X.~Jiang, and F.~Xiao, ``Physical layer authentication utilizing beam pattern features in millimeter-wave {MIMO} systems,'' \emph{{IEEE} Trans. Dependable Secur. Comput.}, pp. 1--15, 2024.

\bibitem{kay1993fundamentals}
S.~M. Kay, \emph{{Fundamentals of Statistical Signal Processing: Estimation Theory}}.\hskip 1em plus 0.5em minus 0.4em\relax Prentice-Hall, Inc., 1993.

\bibitem{liu2009new}
H.~Liu, Y.~Tang, and H.~H. Zhang, ``A new chi-square approximation to the distribution of non-negative definite quadratic forms in non-central normal variables,'' \emph{Computational Statistics \& Data Analysis}, vol.~53, no.~4, pp. 853--856, 2009.

\bibitem{razaviyayn2014successive}
M.~Razaviyayn, ``Successive convex approximation: Analysis and applications,'' Ph.D. dissertation, University of Minnesota, 2014.

\bibitem{DBLP:journals/twc/LiLHCW19}
M.~Li, C.~Liu, S.~V. Hanly, I.~B. Collings, and P.~Whiting, ``Explore and eliminate: Optimized two-stage search for millimeter-wave beam alignment,'' \emph{{IEEE} Trans. Wirel. Commun.}, vol.~18, no.~9, pp. 4379--4393, 2019.

\bibitem{DBLP:journals/tcom/DarsenaV22}
D.~Darsena and F.~Verde, ``Anti-jamming beam alignment in millimeter-wave {MIMO} systems,'' \emph{{IEEE} Trans. Commun.}, vol.~70, no.~8, pp. 5417--5433, 2022.

\bibitem{DBLP:journals/tifs/DinhVanHCFMH23}
S.~Dinh{-}Van, T.~M. Hoang, B.~B. Cebecioglu, D.~S. Fowler, Y.~K. Mo, and M.~D. Higgins, ``A defensive strategy against beam training attack in {5G mmWave} networks for manufacturing,'' \emph{{IEEE} Trans. Inf. Forensics Secur.}, vol.~18, pp. 2204--2217, 2023.

\bibitem{DBLP:journals/icl/XingQCWLL23}
Z.~Xing, C.~Qi, Y.~Cheng, Y.~Wu, D.~Lv, and P.~Li, ``Covert millimeter wave communications based on beam sweeping,'' \emph{{IEEE} Commun. Lett.}, vol.~27, no.~5, pp. 1287--1291, 2023.

\bibitem{DBLP:journals/tce/JuZBLPWO23}
Y.~Ju, G.~Zou, H.~Bai, L.~Liu, Q.~Pei, C.~Wu, and S.~A. Otaibi, ``Random beam switching: {A} physical layer key generation approach to safeguard {mmWave} electronic devices,'' \emph{{IEEE} Trans. Consumer Electron.}, vol.~69, no.~3, pp. 594--607, 2023.

\bibitem{10735492}
J.~Li, L.~Lazos, and M.~Li, ``Secbeam: Securing {mmWave} beam alignment against beam-stealing attacks,'' in \emph{Proc. IEEE Conf. Commun. Netw. Secur. (CNS)}, 2024, pp. 1--9.

\end{thebibliography}

\end{document}